\documentclass[a4paper]{article}
\usepackage[utf8]{inputenc}
\usepackage{geometry}
\usepackage{tikz}

\usepackage[overload]{empheq}

\usepackage{hyperref}

\usepackage{dsfont}
\usepackage{amsmath}
\usepackage{amsthm}
\usepackage{mathtools}
\usepackage{amssymb}
\usepackage{amstext}
\usepackage{amsfonts}
\usepackage{nicefrac}
\usepackage{bbm}
\usepackage{tcolorbox}
\usepackage[thinc]{esdiff}

\usepackage{ stmaryrd }
\usepackage{enumerate}
\usepackage{graphicx}
\usepackage{color}
\usepackage{caption}
\usepackage{subcaption}
\usepackage{adjustbox}
\graphicspath{{figures/}}
\newcommand{\E}{\mathcal{E}}
\newcommand{\A}{\mathcal{A}}
\newcommand{\R}{\mathbb{R}}
\newcommand{\C}{\mathbb{C}}
\newcommand{\N}{\mathbb{N}}
\newcommand{\intr}{\int\limits_{\R^3}}
\newcommand{\di}{\, \mathrm{d}}
\newcommand{\I}[2]{I^{X_2}_{#1,#2}}
\newcommand{\In}{\I{\lambda}{R_{n}}}
\newcommand{\rhon}{\rho^{X_2}_{R_n}}
\newcommand{\rhogn}[1]{\rho_{\gamma^{(#1)}_{n}}}

\newcommand{\gamman}{\gamma^{X_2}_{R_n}}
\renewcommand{\H}{\mathcal{H}}

\newcommand{\gegen}[2]{\xrightarrow{#1 \to #2}}
\newcommand{\limit}[2]{\lim \limits_{#1 \to #2}}
\newcommand{\liminfn}{\liminf\limits_{n \to \infty}}
\newcommand{\limsupn}{\limsup\limits_{n \to \infty}}
\newcommand{\minimum}{\min \limits_{\alpha \in [0,\frac{\lambda}{2}]} \big( I^{X}_\alpha + I^{X}_{\lambda - \alpha} \big)}
\newcommand{\norm}[1]{\left\lVert#1\right\rVert}

\DeclarePairedDelimiterX{\abs}[1]{\lvert}{\rvert}{#1}
\DeclareMathOperator{\tr}{tr}
\DeclareMathOperator{\supp}{supp}
\DeclareMathOperator{\dist}{dist}

\DeclareMathOperator{\kernel}{Ker}
\DeclareMathOperator{\Lap}{\Delta}
\DeclareMathOperator*{\argmin}{arg\,min}

\newcommand{\sech}{\mathrm{sech}}

\usepackage[thinc]{esdiff}
\renewcommand{\epsilon}{\varepsilon}
\renewcommand{\phi}{\varphi}

\newtheorem{theorem}{Theorem}
\newtheorem*{theorem*}{Theorem}
\newtheorem{lemma}{Lemma}
\newtheorem{proposition}{Proposition}
\newtheorem*{remark}{Remark}
\newtheorem{ass}{Assumption}
\usepackage[symbols,nogroupskip,sort=standard]{glossaries-extra}

\glsxtrnewsymbol[description={set of admissible $N$-electron wavefunctions}]{AN}{$\mathcal{A}_N$}

\glsxtrnewsymbol[description={exponents in part 3 of Assumption \ref{assumption1} for  controlling the derivative of $e_{xc}$}]{beta}{$\beta_{\pm}$ }

\glsxtrnewsymbol[description={closed-shell one-electron reduced density operator}]{gamma}{$\gamma$ }

\glsxtrnewsymbol[description={bilinear form associated with the Hartree energy}]{D}{$D[\cdot,\cdot]$ }

\glsxtrnewsymbol[description={set of $N$-body density matrices}]{density_matrices}{$D_N$}

\glsxtrnewsymbol[description={energy functional for the $X$-atom }]{energy_functional}{$\E^X[\cdot]$ }
\glsxtrnewsymbol[description={energy functional for the$X_2$-molecule }]{energy_functional2}{$\E^{X_2}[\cdot]$ }

\glsxtrnewsymbol[description={TFDW energy of the system with mass $\alpha$ }]{ealpha}{$E_{\alpha}$}

\glsxtrnewsymbol[description={quantum mechanical ground state energy}]{ground_state_QM}{$E^{QM}_0$ }

\glsxtrnewsymbol[description={exchange-correlation energy}]{exchange}{$E_{ex}[\cdot ]$}
\glsxtrnewsymbol[description={LDA exchange-correlation function}]{exchange_LDA}{$e_{ex}$}

\glsxtrnewsymbol[description={Levy-Lieb energy functional}]{levy_lieb}{$F_{LL}[\cdot]$}

\glsxtrnewsymbol[description={Hilbert space given by $ \bigwedge_{i=1}^N L^2(\R^3)$ }]{hilbert2}{$\mathcal{H}_N$}

\glsxtrnewsymbol[description={subspace of $\mathfrak{S}_1$ with finite kinetic energy}]{hilbert_space}{$\mathcal{H}$}

\glsxtrnewsymbol[description={hamiltonian coming from the Euler-Lagrange equation }]{hamiltonian2}{$h_\gamma$ }

\glsxtrnewsymbol[description={$N$-electron hamiltonian with potential $V$}]{hamiltonian}{$H^V_N$}

\glsxtrnewsymbol[description={energy of the $X$-atom with $\lambda$ electrons surrounding it}]{Ilambda}{$I_{\lambda}^X$}

\glsxtrnewsymbol[description={energy of the $X_2$-molecules with $\lambda$ electrons surrounding it and distance $R$ between the nuclei}]{IlambdaX2}{$I_{\lambda,R}^{X_2}$ }
\glsxtrnewsymbol[description={energy of the problem at infinity, i.e.~without potential, with $\lambda$ electrons surrounding it}]{I_inf}{$I_{\lambda,}^{\infty}$ }

\glsxtrnewsymbol[description={Hartree energy of a density}]{J}{$J[\cdot]$ }

\glsxtrnewsymbol[description={$N$-body wavefunction}]{wavefunction}{$\Psi$}

\glsxtrnewsymbol[description={set of admissible density matrices $\gamma$ with mass $\tr[\gamma] =\lambda$}]{konvex_set}{$K_\lambda$}

\glsxtrnewsymbol[description={one-body reduced electron density}]{rho}{$\rho$}

\glsxtrnewsymbol[description={set of admissible one-body electron densities arising from the wavefunctions in $\A_N$}]{RN}{$\mathcal{R}_N$}

\glsxtrnewsymbol[description={set of admissible one-body electron densities arising from the density matrices in $D_N$}]{RDN}{$\mathcal{R}D_N$}

\glsxtrnewsymbol[description={space of position and spin; elements are denoted by $z = (x,y)$}]{position_and_spin}{$\R^3_\Sigma$}

\glsxtrnewsymbol[description={set of trace class operators on $L^2(\R^3)$}]{trace_class}{$\mathfrak{S}_1$}

\glsxtrnewsymbol[description={set of spin-states $\big \lbrace \lvert \uparrow \rangle, \lvert \downarrow \rangle \big \rbrace$}]{spin}{$\Sigma$}

\glsxtrnewsymbol[description={kinetic energy of the system, depending on the setting of $\rho, \gamma $ or $\Psi$}]{kinetic}{$T[\cdot]$}
\glsxtrnewsymbol[description={unitary translation operator }]{tau}{$\tau_R$ }

 \glsxtrnewsymbol[description={eigenvalues corresponding ti $h_\gamma$ }]{theta}{$\theta_l$ }

\glsxtrnewsymbol[description={Coulomb potential generated by the clamped nuclei}]{V-potential}{$V$}

\glsxtrnewsymbol[description={electron-electron interaction energy}]{vee}{$V_{ee}[\cdot]$}
\glsxtrnewsymbol[description={electron-nuclei interaction energy}]{vne}{$V_{ne}[\cdot]$}

\glsxtrnewsymbol[description={partition of unity in the dichotomy case }]{xi}{$\xi$ }
\glsxtrnewsymbol[description={partition of unity in the dichotomy case }]{zeta}{$\zeta$}

\usetikzlibrary{shapes,arrows}
\title{Dissociation limit in Kohn-Sham density functional theory}
\author{
S\"oren Behr \\
Department of Mathematics \\
 Technische Universit\"at M\"unchen \\
\small{behr@ma.tum.de}
\\[3mm]
Benedikt R. Graswald \\
Department of Mathematics \\
 Technische Universit\"at M\"unchen \\
\small{graswabe@ma.tum.de}
}

\begin{document}

\maketitle

\begin{abstract}
    We consider the dissociation limit for molecules of the type $X_2$ in the Kohn-Sham density functional theory setting, where $X$ can be any element with $N$ electrons. We prove that when the two atoms in the system are torn infinitely far apart, the energy of the system convergences to $\min \limits_{\alpha \in [0,N]} \big( I^{X}_{\alpha} + I^{X}_{2N-\alpha} \big)$, where $I^{X}_{\alpha}$ denotes the energy of the atom with $\alpha$ electrons surrounding it.
    Depending on the ``strength'' of the exchange this minimum might not be equal to the symmetric splitting $2I^{X}_{N}$.
    We show numerically that for the $H_2$-molecule with Dirac exchange this gives the expected result of twice the energy of a H-atom $2 I^{H}_1$.
\end{abstract}{}

\tableofcontents

\section{Introduction}

Density functional theory (DFT) was developed by Hohenberg, Kohn and Sham \cite{hohenbergkohn64, kohnsham65} in the 1960s and is to this day one of the most widely spread electronic structure models in quantum chemistry, biology and materials science because of its good compromise between accuracy and computational cost.
The idea behind DFT is to transform the high-dimensional Schr\"odinger equation into a low-dimensional and thus computationally manageable problem.

The trade-off in this approach is the introduction of the so-called exchange-correlation functional, which is in theory exact but in practice unknown.
Therefore a lot of effort \cite{perdew_zunger, perdew_wang, becke88,b3lyp,pbe} has gone into building good approximations to this functional.
In this paper we consider the simplest form of these models, the local density approximation (LDA) first proposed in \cite{kohnsham65}, some standard references are \cite{parr_young, engel2013density}. Even here the resulting mathematical properties are still far from being well understood.
Furthermore as observed in reference \cite{Medvedev49} starting in the early 2000s newer  approximations actually become worse in predicting the electron densities. This is due to only focusing on the energies and in the process sacrificing mathematical rigor in favor of the flexibility of fitting to empirical data. 
Thus in the present article we want to focus on fundamental properties that the exchange-correlation functional should fulfill.

Our main goal is to analyze the  dissociation limit of any symmetric diatomic molecule, i.e.~any molecule of the form $X_2$, in Kohn-Sham (KS) DFT. Simply put we ask the question, what  happens to the energy of the system, when the distance between the two atoms is artificially increased further and further until they are torn infinitely far apart?
Our main result takes the following form
\begin{theorem*}[Theorem \ref{thm:main} -- Informal Version] 
Let $I^{X_2}_{2N, R}$ and $I^{X}_\lambda$ be the energy of the $X_2$-molecule with distance $R$ between the atoms and the $X$-atom with $\lambda$ electrons, respectively, defined by \eqref{eq:ground_state_molecule}. Then we have
\begin{equation} \label{eq:splitting}
    \lim\limits_{R \to \infty} I_{2N, R}^{X_2} = \min \limits_{\alpha \in [0,N]} \big( I^{X}_\alpha + I^{X}_{2N - \alpha} \big). 
\end{equation}{}
\end{theorem*}
In the long range limit, the ground-state energy of the $X_2$-molecule is identical to the energy of two non-interacting atoms - one with electron mass $\alpha$ and one with electron mass $2N - \alpha$. The physical expectation here is, that it is optimal to split the electrons evenly (i.e. the minimum is attained for $\alpha = N$).

The question if or rather for which $\lambda$ one has symmetric splitting, i.e. given a family of infima $I_\lambda$ with mass $\lambda$ if 
\[
2 I_{\lambda} < I_{\lambda + \epsilon} + I_{\lambda - \epsilon} \quad \text{for all } 0 < \epsilon < \lambda ,
\]
already plays an important role in Thomas-Fermi and related theories, see e.g.~\cite{lieb1997thomas}.

To our knowledge the fact that the lowest energy splitting is always given by two neutral atoms is not even proven in full quantum mechanics, rather only in Thomas-Fermi theory and perturbations thereof, where the behaviour of the energy with respect to the particle number is completely understood. A simple sketch of this is presented in Figure \ref{fig:plot_energy_TF}. 

\begin{figure}[h!]
    \centering
    \includegraphics{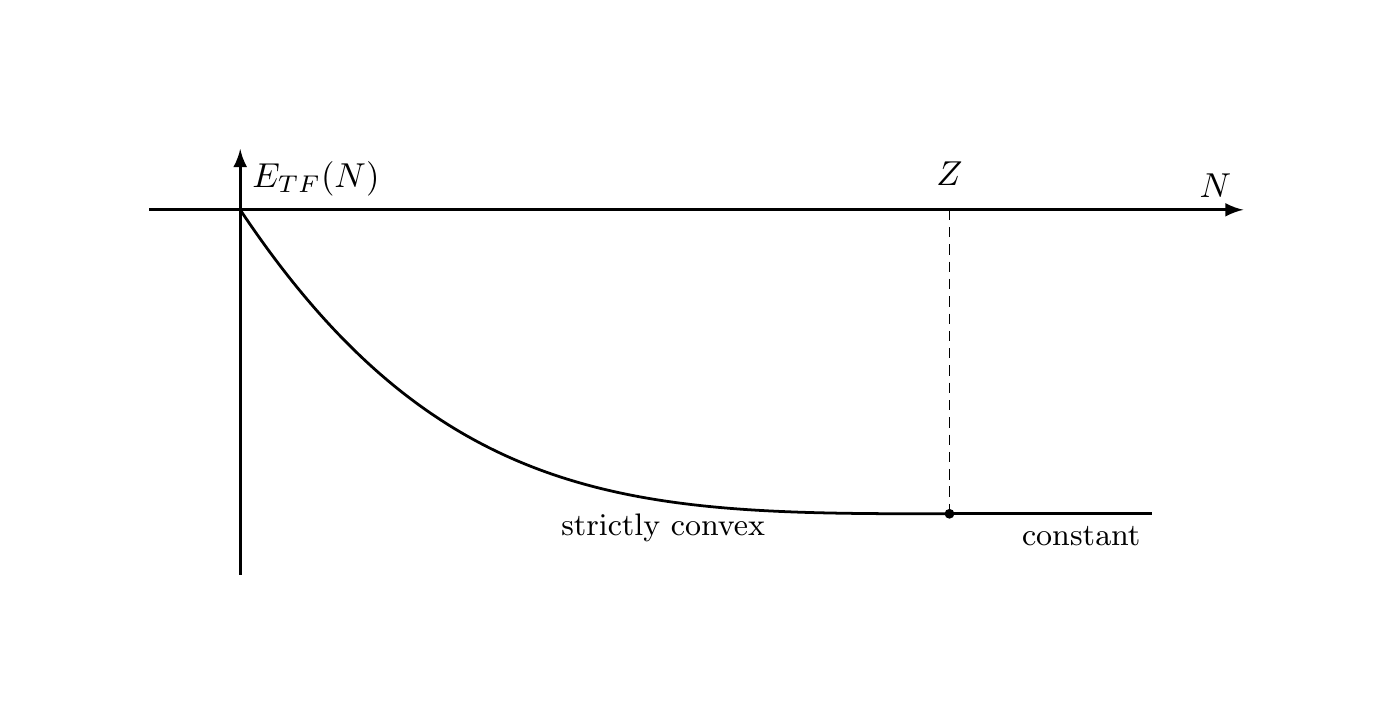}
    \caption{The Thomas-Fermi energy $E_{TF}(N)$ with respect to the particle number $N$. For positively charged systems $N<Z$ it is strictly convex, for $N>Z$ it remains constant.}
    \label{fig:plot_energy_TF}
\end{figure}

Note furthermore that in full quantum mechanics and thus for exact HK-DFT charge quantization occurs, i.e. $\alpha$ in \eqref{eq:splitting} can be restricted to integer values, as proven in \cite{friesecke2003multiconfiguration}.

As will be discussed in detail in Section \ref{sec:dissociation} if the exchange becomes too strong, we observe symmetry breaking, i.e.~the right hand side of \eqref{eq:splitting} does not equal 2$I^{X}_N$.

In the physics literature, this is a well-known challenge: While spin-restricted Kohn-Sham calculations yield qualitatively correct results (i.e. by nature preserve spin-symmetry) they only give reasonable energies close to the actual bond length. Spin unrestricted schemes on the other hand yield better energies but may prefer ionic solutions at long ranges \cite{fuchs2005,perdew1995}. 

This dilemma has recently attracted mathematical interest. In case of the \(H_2\) molecule at fixed bond-length (see \cite{lu_symmetrybreaking}) and for periodic systems (see \cite{gontier2020}), symmetry breaking occurs for  sufficiently strong exchange contributions.
These issues in LDA-DFT and related theories like Thomas-Fermi-Dirac-von Weizsäcker is caused by the Dirac term $- \int \rho^{\nicefrac{4}{3}}$, which to some extend makes the functional concave and can thus lead also to nonattainment, see e.g. \cite{lu_otto}.

The rest of the paper is structured as follows:
The next section sets the stage by defining and motivating all the energy functionals needed, giving our main result in Theorem \ref{thm:main} the necessary detail.
In Section 3 we put it into context by  considering first a one-dimensional DFT model where we can always determine the right hand side of \eqref{eq:splitting}.
Then we consider the full three dimensional case and fill the gap in our theoretical results by numerical evidence.

The last section contains all the proofs, with the most interesting point being that we apply the concentration-compactness lemma not to a minimizing sequence but to a sequence of minimizers. Figure \ref{fig:structure_proof} summarizes the structure of the proof to help not get lost in technical details. 

\section{Setting the stage} 
\subsection{Density functional theory} \label{sec:dft}

To put our result into perspective we recall here shortly the basic fundamentals of DFT. A  standard reference would be \cite{parr_young}.  Readers familiar with the topic might want to skip this section. 
As a quick reference guide for the notation in use, we created a list of symbols at the end of this paper, which the reader might want to consult from time to time. 

The starting point is a system in Born-Oppenheimer approximation \cite{bornoppenheimer, gustafson2003mathematical}, i.e.~a system of $N$ non-relativistic electrons under influence of an external potential $V(x)$ and with a repulsive interaction potential $v_{ee}(x-y)$.

For a molecule with $M$ atomic nuclei at positions $R_1, \ldots, R_M \in \R^{3}$, with individual charges $Z_1, \ldots, Z_M \in \N$ and total atomic charge $Z = \sum \limits_{i =1}^{M} Z_{i}$, and with $N$ electrons the potential $v(x)$ is just the ensuing Coulomb potential of their positions and charges 
 \begin{align*}
 V(x) := - \sum_{i = 1}^{M} \frac{Z_i}{|x - R_i|}, \quad x \in \R^3.
 \end{align*}
The class of admissible functions $\A_N$ -- the so-called $N$-electron wave functions --  is given by 
 \begin{gather*}
 \mathcal{A}_N := \left\lbrace \Psi \in L^{2} \left( (\R^3 \times \Sigma)^{N}  ; \C \right) : \nabla \Psi \in L^{2}, \Psi \text{ antisymmetric} , ||\Psi||_{L^2} =1 \right\rbrace,\\
 \text{where } \enskip \Sigma :=\big \lbrace \lvert \uparrow \rangle, \lvert \downarrow \rangle \big \rbrace \enskip \text{denotes the set of spin-states.}
 \end{gather*}
 In the following we denote the space of position and spin by $\R^3_\Sigma = \R^3 \times \Sigma $ and write $z_i = (x_i,s_i) \in R^3_\Sigma$ for the pair of position and spin of the $i^{\text{th}}$ particle. 
Now we can finally define the quantum mechanical energy functional $\E^{QM}$,
\begin{equation}
\E^{QM}[\Psi] := T[\Psi] + V_{ne}[\Psi] + V_{ee}[\Psi], \label{energy_functional}
\end{equation} 
where
\linespread{1.125}
\begin{align*}
T[\Psi] &:= \frac{1}{2} \int_{(\R^{3}_{\Sigma})^{N}} \sum_{i=1}^{N} \left| \nabla_{x_{i}} \Psi(x_1,s_1, \ldots, x_{N}, s_N) \right|^{2} \di z_{1} \ldots \di z_{N} 
\end{align*}
describes the kinetic energy,
\begin{align*}
V_{ne}[\Psi] &:= \int_{(\R^{3}_{\Sigma})^{N}} \sum_{i=1}^{N} V(x_i) \left| \Psi(x_1,s_1, \ldots, x_{N}, s_N) \right|^{2} \di z_{1}\ldots \di z_{N}
\end{align*}
gives the electron-nuclei interaction energy, and
\begin{align*}
V_{ee}[\Psi] &:= \int_{(\R^{3}_{\Sigma})^{N}} \sum_{1 \leq i <j \leq N} v_{ee}(x_i-x_j) \left| \Psi(x_1,s_1, \ldots, x_{N}, s_N) \right|^{2} \di z_{1}\ldots \di z_{N} \\
\end{align*}
is the electron-electron interaction energy. Here we used the notation $\int_{\R^3_\Sigma} f(z) \di z = \sum_{s \in \Sigma} \intr f(x,s) \di x$.
\linespread{1.4}
The exact quantum mechanical ground state energy is now defined as 
\begin{align} \label{exactgroundstateenergy}
E_{0}^{QM} : = \inf \limits_{\Psi \in \A_N} \E^{QM}[\Psi]. 
\end{align}

Unfortunately  due to the curse of dimensionality there is no hope of ever solving \eqref{exactgroundstateenergy} for interesting molecular systems. 
This is where a central result going back to Hohenberg and Kohn \cite{hohenbergkohn64} comes into play. We state it in the more modern formulation due to Levy and Lieb \cite{levy1979universal, lieb2002density}, see also \cite{chen_friesecke}:
The quantum mechanical ground state energy $E^{QM}_0$ \eqref{exactgroundstateenergy} only depends on the one-body density $\rho$ given by
\begin{align}\label{eq:onebodydensity}
\rho(x) = \sum_{s \in \Sigma} \int_{\R_\Sigma^{3(N-1)}} |\Psi(x,s,z_2,\ldots, z_N)|^2 \di z_2 \ldots \di z_N.
\end{align}
Furthermore it can be recovered exactly by the following minimization
\begin{align}\label{eq:levylieb}
    E^{QM}_0 = \inf_{\rho \in \mathcal{R}_N} \left( F_{LL}[\rho] + \intr v(x) \rho(x) \di x \right),
\end{align}
where the functional $F_{LL}$ is given by
\begin{align}\label{eq:levyliebfunctional}
F_{LL}[\rho] = \min_{\Psi \in \A_N, \Psi \mapsto \rho} \left( T[\Psi] + V_{ee}[\Psi]\right).
\end{align}
Here the map $\Psi \mapsto \rho$ describes the relationship in \eqref{eq:onebodydensity}, i.e.~ $\Psi$ has one-body density $\rho$ and $\mathcal{R}_N$ denotes the set of admissible densities $\rho$ arising via \eqref{eq:onebodydensity} from the set of admissible wavefunctions $\A_N$.
 Note that due to \cite{lieb2002density} $\mathcal{R}_N$ has an explicit form
 \begin{align}\label{eq:admissible_set_densities}
     \mathcal{R}_N = \left\lbrace \rho:\R^3 \to \R \bigg\lvert~ \rho \geq0, ~\sqrt{\rho} \in H^1\big(\R^3 \big), \intr \rho(x) \di x = N  \right\rbrace
 \end{align}
and provided $\rho \in \mathcal{R}_N$ the minimum in \eqref{eq:levyliebfunctional} is attained.

The problem one now faces is that there is no tractable expression of $F_{LL}$ which could be used in practice. 
In physics one usually takes the starting point of splitting $F_{LL}[\rho]$ into three parts
\[
F_{LL}[\rho] = T[\rho] + V_{ee} [\rho] +  E_{xc}[\rho],
\]
where $T$ describes a kinetic part and $V_{ee}$ an interaction part and the exchange-correlation $E_{xc}$ contains all the
other terms ensuring that equality holds.
There are of course different choice for the individual functionals, but a particularly successful one has been proposed by Kohn and Sham \cite{kohnsham65}. They came up with the idea to construct the kinetic term $T$ by considering a non-interacting reference system with the same density $\rho$ described by single-particle orbitals $(\varphi_i)_{i=1}^N$ given by
\begin{align}\nonumber
    T[\rho] = T_{KS}[\rho] =
    \min \bigg\{
    \frac{1}{2} \sum_{i=1}^N \int_{\R^3_\Sigma} |\nabla \varphi_i|^2 (z) \di z ~\bigg \lvert~ \varphi \in H^1(\R^3_\Sigma), 
    \int_{\R^3_\Sigma} \overline{\varphi_i(z)}  \varphi_j (z) \di z &= \delta_{ij},  \\
     \sum_{i=1}^N \sum_{s \in \Sigma} |\varphi_i(x,s)|^2 &= \rho(x) \bigg\}.
\end{align}
Note that these orbitals coming from the fictitious non-interacting system, are only connected to the real system by having the same density, a direct interpretation while sometimes loosely done in practice is not theoretically justified.
The interaction term is modeled by an independence ansatz
\[
V_{ee}[\rho] = \intr \intr \frac{\rho(x)\rho(y)}{|x-y|} \di x \di y.
\]
Thus the challenge becomes finding an accurate approximation for $E_{xc}[\rho]$. There is a huge variety of different exchange-correlation functionals (see e.g~\cite{perdew_wang, perdew_zunger, pbe}), each with its advantages and disadvantages.
In the following we will we working with the so-called local density approximation (LDA), meaning that the exchange correlation functional is assumed to be of the following form
\begin{align}\label{eq:exchange_correlation_functional_lda_form}
    E_{xc}[\rho] = \intr e_{xc}(\rho(x)) \di x,
\end{align}
where the function $e_{xc}:\R \to \R$ has to fulfill certain properties. In our case they will be
specified in Section \ref{sec:dissociation} under Assumptions \ref{assumption1}.

The prototypical example for an $E_{xc}[\rho]$--approximation stems from considering the homogeneous electron gas. It goes  back to Dirac \cite{dirac1930note} (for a mathematical derivation see \cite{friesecke1997pair}) and is given by
 \begin{align} \label{eq:dirac_exchange}
     E_{xc}[\rho] = \intr e_{xc}(\rho(x)) \di x,  \quad e_{xc}(\rho) = -
     c_{xc} \rho^{\nicefrac{4}{3}}. 
\end{align}{}

Employing the above ansatz for the different parts of $F_{LL}$, in particular using the orbitals $\Phi = (\varphi_1,\ldots,\varphi_N)$ of the KS-ansatz,  equation \eqref{eq:levylieb} takes the form

\begin{align*}\label{eq:energy_orbitals}
    E^{QM}_0  \approx E^{LDA}_0= \inf \bigg\{ 
    \frac{1}{2} \sum_{i=1}^N \int_{\R^3_\Sigma} |\nabla \varphi_i|^2(z) \di z
    + \intr V(x) \rho(x) \di x + \intr \intr \frac{\rho(x) \rho(y)}{|x-y|} \di x \di y + \intr e_{xc}(\rho(x) ) \di x 
    ~\bigg\lvert~ \\
    \varphi_i \in H^1(\R^3_\Sigma ),~ \int_{\R^3_\Sigma} \overline{\varphi_i(z)} \varphi_j (z) \di z = \delta_{ij},~ \rho(x) = \sum_{i=1}^N \sum_{s \in \Sigma } |\varphi_i(x,s)|^2
    \bigg\}.
\end{align*}

\subsection{Mixed-states}
This section shortly recalls the description of the above problem using mixed states, i.e.~density matrices. For more details see e.g.~\cite{cances}.\\
Let $\mathfrak{S}_1$ denote the vector space of trace class operators on $L^2(\R^3)$ and introduce the subspace $\H:= \{ \gamma \in \mathfrak{S}_1 : |\nabla| \gamma |\nabla| \in \mathfrak{S}_1 \}$ endowed with the norm $|| \cdot ||_{\H} := \tr(|\cdot|) + \tr\big(\big| |\nabla| \cdot |\nabla| \big|\big)$ and the convex set 
\begin{align}
    \label{eq:definition_set_k}
K := \{ \gamma \in \mathcal{S}(L^2(\R^3)) : 0 \leq \gamma \leq 1, \tr(\gamma) < \infty, \tr\big(  |\nabla| \gamma |\nabla|  \big) < \infty \},
\end{align}
where $\mathcal{S}(L^2(\R^3))$ denotes the space of bounded self-adjoint operators on $L^2(\R^3)$.
Next, let us remark that 
\begin{align}
E^{QM}_{0} &= \inf \bigg \lbrace  \langle \Psi \mid H_{N}^{V} \mid \Psi \rangle  \enskip : \enskip \Psi \in \A_N  \bigg \rbrace \label{wave}\\
&= \inf \bigg \lbrace  \tr \big(H_{N}^{V}\Gamma \big)  \enskip : \enskip \Gamma \in D_{N}  \bigg \rbrace, \label{matrices}
\end{align}
where $H_N^V$ is the electronic hamiltonian 
\begin{align}\label{eq:hamiltonian}
    H_N^V : = - \frac{1}{2} \sum_{i=1}^N \Delta_{x_i } - \sum_{i=1}^N V(x_i) + \sum_{1\leq i<j\leq N} \frac{1}{|x_i - x_j|},
\end{align}
and $D_{N}$ is the set of $N$-body density matrices defined by
\begin{align} \label{nbody_density_matrices}
D_N &= \bigg \lbrace \Gamma \in \mathcal{S}\big(\H_{N} \big) : 0 \leq \Gamma \leq 1, \tr(\Gamma) =1, \tr(- \Delta \Gamma) < \infty \bigg \rbrace.
\end{align}
In the above expression, $\mathcal{S}\big(\H_N \big)$ denotes the vector space of bounded self-adjoint operators on the Hilbert space $\H_N,$ where  
\begin{align*}
    \H_N = \bigwedge_{i=1}^N L^2\big( \R^3_\Sigma \big),
\end{align*}
endowed with the standard inner product
\[
\langle \Psi \lvert \Psi' \rangle_{\H_N} = \int_{(R^3_\Sigma )^N}  \overline{\Psi(z_1, \ldots, z_N)} \Psi'(z_1,\ldots,z_N) \di z_1 \ldots \di z_N.
\]
Furthermore the condition $0 \leq \Gamma \leq 1$ stands for $0 \leq \langle \Psi \mid \Gamma \mid \Psi \rangle\leq ||\Psi||^{2}_{\H_N}$ for all $\Psi \in \H_N$.

From a physical point of view, \eqref{wave} and \eqref{matrices} mean that the ground state energy can be computed either by minimizing over pure states -- characterized by wave functions $\Psi$ -- or by minimizing over mixed states -- characterized by density operators $\Gamma$.

As before we define the electronic density for any $N$-electron density operator $\Gamma \in D_N$
\begin{align}
\rho_{\Gamma}(x) : = N \sum_{\sigma \in \Sigma} \int_{(\R^{3}_{\Sigma})^{(N-1)}} \Gamma(x,\sigma, z_{2}, \ldots, z_{N}; x, \sigma, z_{2}, \ldots, z_{N}) \di z_{2} \ldots \di z_{N}.
\end{align}
Note that here and below we use the same notation for an operator and its Green kernel.

Then we get for the electron densities
\begin{align*}
\bigg \lbrace 
\rho : \R^3 \to \R :
 \exists \Gamma
 \in D_N, \rho_{\Gamma} = \rho \bigg \rbrace 
 = \mathcal{R}_N
 = \bigg \lbrace \rho : \R^3 \to \R : \rho \geq 0, \sqrt{\rho} \in H^{1}(\R^3), \intr \rho \di x = N \bigg \rbrace.
\end{align*}

Let $\Gamma \in D_N$ be in the set of $N$-body density matrices, then the one-electron reduced density operator $\Upsilon_{\Gamma}$ associated with $\Gamma$ which is the self-adjoint operator on $L^{2}(\R^{3}_{\Sigma})$ with kernel 
\begin{align*}
\Upsilon_{\Gamma}(x,s;y,t) = N \int_{(\R^{3}_{\Sigma})^{N-1}} \Gamma(x,s, z_{2}, \ldots, z_{N}; y, t,z_{2}, \ldots, z_{N}) \di z_{2} \ldots \di z_{N}.
\end{align*}

Furthermore it is known, see e.g. \cite{reduced}, that
\begin{align}
\bigg \lbrace \Upsilon : \exists  \Gamma \in D_{N}, \rho_{\Gamma} = \rho \bigg \rbrace = 
\bigg \lbrace \Upsilon \in \mathcal{R} D_{N}:  \rho_{\Gamma} = \rho \bigg \rbrace, \label{upsilon}
\end{align}
where 
\begin{align}
\mathcal{R} D_{N} &= \bigg \lbrace \Upsilon \in S( L^{2}(\R^{3}_{\Sigma})) : 0 \leq \Upsilon \leq 1, \tr(\Upsilon) =N, \tr(- \Delta_{x} \Upsilon) < \infty \bigg \rbrace  \quad
\text{and} \\
\rho_{\Upsilon}(x) &:= \sum_{\sigma \in \Sigma} \Upsilon(x, \sigma; x, \sigma).
\end{align}

This leads to the so-called extended Kohn-Sham models
\begin{align}
I^{EKS}_{N}[V] : = \inf \bigg \lbrace \tr \big(- \tfrac{1}{2} \Delta_{x} \Upsilon  \big) + \intr \rho_{\Upsilon} V \di x + J[\rho_{\Upsilon}] + E_{ex}[\rho_{\Upsilon}] : \Upsilon \in \mathcal{R}D_{N} \bigg \rbrace. \label{extendedkohnsham}
\end{align}
Note, up to now no approximation has been made, such that for the exact exchange-correlation functional $E^{QM}_{0} = I^{EKS}_{N}$ for any molecular system containing $N$ electrons.
Unfortunately as mentioned in Subsection \ref{sec:dft}, there is no tractable expression of $E_{xc}[\rho]$ that can be used in numerical simulations.

Before proceeding further, and for the sake of simplicity, we will restrict ourselves to closed-shell, spin-unpolarized systems.
This means that we will only consider molecular systems with an even number of electrons $N = 2 N_{p}$, where $N_{p}$ is the number of electron pairs in the system, and we will assume that electrons ``go by pairs''.

Hence, the constraints on the one-electron reduced density operator originating from the closed-shell approximation read:
\begin{align}
\Upsilon(x, \lvert \uparrow \rangle, y, \lvert \uparrow \rangle) = \Upsilon(x, \lvert \downarrow \rangle, y, \lvert \downarrow \rangle)
 \enskip \text{and} \enskip 
 \Upsilon(x, \lvert \uparrow \rangle, y, \lvert \downarrow \rangle) = \Upsilon(x, \lvert \downarrow \rangle, y, \lvert \uparrow \rangle) = 0.
\end{align}
Introducing $\gamma(x,y) = \Upsilon(x, \lvert \uparrow \rangle, y, \lvert \uparrow \rangle)$ and denoting $\rho_{\gamma}(x) = 2 \gamma(x,x)$, we obtain the spin-unpolarized (or closed-shell or restricted) extended Kohn-Sham model

\begin{gather}
I_{N}^{REKS}(V) = \inf \bigg \lbrace \E(\gamma) : \gamma \in K_{N_{p}} \bigg \rbrace, 
\end{gather}
where the energy functional $\E$ is given by
\begin{gather}
\E(\gamma) = \tr(-\Delta \gamma) + \intr \rho_{\gamma} V \di x + J[\rho_{\gamma}] + E_{xc}[\rho_{\gamma}],
\end{gather}
and the admissible set looks like
\begin{gather}
K_{N_{p}} = \bigg \lbrace \gamma \in \mathcal{S}(L^{2}(\R^{3})) : 0 \leq \gamma \leq 1, \tr(\gamma) = N_{p}, \tr(- \Delta \gamma) < \infty \bigg \rbrace.
\end{gather}

Note that the factor $\tfrac{1}{2}$  in front of the kinetic part of $H^V_N$ from \eqref{eq:hamiltonian}
 vanishes here in front of the trace due to the definition of $\gamma$ and accounts for the spin.

Furthermore, by spectral theory we have for any $\gamma \in K_{N_{p}}$
\begin{gather}
\gamma = \sum_{i \geq 1} \lambda_{i} \lvert \phi_i \rangle \langle \phi_{i} \rvert 
\end{gather}
with
\begin{gather} 
\phi_{i}  \in H^{1}(\R^3), \quad \intr \phi_{i} \phi_{j} \di x = \delta_{ij}, \quad \lambda_{i} \in [0,1],~ \sum_{i=1}^{\infty} \lambda_{i}= N_{p}, \quad \sum_{i =1}^{\infty} \lambda_{i} ||\nabla \phi_{i}||_{L^{2}}^{2} < \infty.
\end{gather}

\subsection{Dissociation}
In this section we shortly introduce the energy functionals we will be using in this paper.
The Kohn-Sham energy functional  is given by 
\begin{equation} \label{energy_molecule}
   \E^{V} [\gamma] := \tr[-\Lap \gamma] +  \intr V \rho \di x + \frac{1}{2} \intr \intr \frac{\rho(x) \rho(y)}{|x-y|} \di x \di y  +  \intr e_{xc}\big(\rho(x)\big) \di x ,
\end{equation}{}
where $ \rho(x) = 2 \gamma(x,x)$ and $V$ denotes the external potential. Note that the factor 2 is used since we are considering a spin-unpolarized system.
Let $X$ be any atom with $Z$ number of protons. Then for the $X_2$ molecule we have 
\begin{align}
\label{eq:energy_functional_molecule}
V^{X_2}_{R}
= - \frac{Z}{|\cdot|} - \frac{Z}{|\cdot - R|},
\quad \E^{X_2}_R[\gamma] := \E^{V^{X_2}_R}[\gamma],
\end{align}
and similar for the $X$-atom
\begin{align}
\label{eq:energy_functional_atom}
 V^{X}= - \frac{Z}{|\cdot|}, 
\quad \E^{X}[\gamma] := \E^{V^{X}}[\gamma].
\end{align}
Here and in the following to keep notation a bit simpler we will denote by $R$ the position of the second nucleus and also its distance to the origin, as long as it is clear from context which one we are referring to.

We then define the ground state energies
\begin{equation} \label{eq:ground_state_molecule}
I^{X_2}_{\lambda, R} : = \inf\limits_{\gamma \in K_\lambda}
\E^{X_2}_{R}[\gamma],
\quad
I^{X}_{\lambda} : = \inf\limits_{\gamma \in K_\lambda}
\E^{X}[\gamma],
\end{equation}
where the admissible set is given by
\begin{equation} 
K_\lambda := \big\{ \gamma \in S(L^2(\R^3)): 0 \leq \gamma \leq 1, \tr(\gamma) = \lambda, \tr(- \Lap \gamma) < \infty \big\}.
\end{equation}
Furthermore we introduce the problem at infinity, corresponding to a system without nuclei 
\begin{equation} \label{eq:problem_at_infinity}
    I^{\infty}_{\lambda} : = \inf\limits_{\gamma \in K_\lambda} \E^{ \infty}[\gamma], \quad
    \E^{\infty}[\gamma] :=  \tr[-\Lap \gamma]  + \frac{1}{2} \intr \intr \frac{\rho(x) \rho(y)}{|x-y|} \di x \di y  +  \intr e_{xc}\big(\rho\big) \di x.
\end{equation}{}

To shorten notation we will denote by 
$T[\gamma] = \tr[-\Lap \gamma]$  
the kinetic energy, 
$V[\gamma] = \int V(x) \rho(x)$  describes the electron-nuclei interaction and the exchange-correlation term is given by $E_{xc}[\rho] = \int e_{xc}(\rho) \di x$.
Furthermore  the Hartree energy is given by
\begin{align*}
J[\rho] =\tfrac{1}{2} \intr \intr \frac{\rho(x) \rho(y)}{|x-y|} \di x \di y
\end{align*}
with its corresponding bilinear form $D[f,g]$ being
\begin{align*}
     D[f,g] =  \intr \intr \frac{f(x) g(y)}{|x-y|} \di x \di y.
\end{align*} 
Furthermore if a certain statement holds true for all of the three infima, we will sometimes simply write it holds for the map $\lambda \mapsto I_\lambda$.

Next let us give the assumption on the exchange-correlation term. Note that we can use the same setting as \cite{cances} for the local-density approximation (LDA).
\begin{ass}[LDA-exchange-correlation] \label{assumption1}
Let $e_{xc} :\R_{+} \to \R$ be a $C^1$-function such that 
\begin{enumerate}
    \item $e_{xc}(0) = 0$,
    \item $e_{xc}' \leq 0$,
    \item $\exists 0< \beta_{-} \leq \beta_{+} < \frac{2}{3}$ such that $\big|e'_{xc}(\rho )\big| \leq C \big(\rho^{\beta_{-}} + \rho^{\beta_{+}}\big)$,
    \item $\exists 1 \leq \alpha < \frac{3}{2}$ such that $\limsup \limits_{\rho \to 0}
    \frac{e_{xc}(\rho)}{\rho^{\alpha}} < 0$.
\end{enumerate}{}
\end{ass}{}
Note that the prototypical exchange-correlation functional in the LDA-setting \eqref{eq:dirac_exchange} coming from the uniform electron gas satisfies these assumptions with $\alpha = \tfrac{4}{3}$ and $\beta_{-}=\beta_{+}= \tfrac{1}{3}$.



\begin{remark}
The existence of minimizers to these functionals for neutral or positively charged systems is due to \cite{cances}.
We will also be using the following standard results proven there, which we summaries in the Lemmata \ref{lem:subadditivity} -- \ref{lem:bounds}.
\end{remark}

First some properties of the electron mass to ground state energy map $\lambda \mapsto I_\lambda$.
\begin{lemma}[Properties of the infimum \cite{cances}] \label{lem:subadditivity}
Let $I^{X_2}_{\lambda,R},I^{X}_\lambda$ and $I^{\infty}_\lambda $ be as defined above. For the molecular energy assume $R$ is fixed, but arbitrary.
Then the following holds 

\begin{enumerate}
\item All three maps $\lambda \mapsto I^{X_2}_{\lambda,R},~ \lambda \mapsto I^X_\lambda$ and $\lambda \mapsto I^\infty_\lambda$ are continuous and strictly decreasing for any $\lambda$ in the domain $\lambda \in [0,\infty)$.
\item We always have $I^{X_2}_{0,R} = I^X_0 = I^\infty_0= 0$ and $ ~ -\infty < I^{X_2}_{R,\lambda} < I^{X}_\lambda < I^{\infty}_{\lambda}  < 0 $ for $\lambda>0$.
\item Furthermore all three maps  satisfies the subadditivty condition, i.e. for any of the maps denoted by $\lambda \mapsto I_\lambda$ we have
\begin{equation} \label{subadditivity}
I_{\lambda} \leq I_{\alpha} + I^{\infty}_{\lambda - \alpha}~  \forall \alpha \in [0,\lambda]
\end{equation}
\end{enumerate}
\end{lemma}

Furthermore the next lemma says that minimizing sequences of our problems cannot vanish in the sense of \cite{lions_cc}.

\begin{lemma}[Non-vanishing \cite{cances}]
Let $\lambda >0 $ and $(\gamma_n)_n$ a minimizing sequence for any of the problems \eqref{eq:ground_state_molecule} or \eqref{eq:problem_at_infinity}. Then the sequence $(\rho_{\gamma_n})_n$ cannot vanish in the sense of \cite{lions_cc}, which means that
\[
\exists R>0: ~ \text{such that } \limit{n}{\infty} \sup\limits_{x \in \R^3} \int_{B_R(x)} \rho_{\gamma_n} (x) \di x >0.
\]
\end{lemma} 

Additionally we remark the classical continuity properties.
\begin{lemma}[Continuity \cite{cances}] \label{lem:continuity_E}
The three functionals
\(
\E^{X_2}, \E^{X}, \E^{\infty}
\)
are all continuous on the space $\H =\{ \gamma \in \mathfrak{S}_1 : |\nabla| \gamma |\nabla| \in \mathfrak{S}_1 \}$.
\end{lemma}{}

The next lemma summarizes the standard bounds on the energy functional. We note that in the following C will describe a generic constant, which may have different values
at each appearance, indicating some finite positive constant independent of the surrounding
variables.
\begin{lemma}[Bounds on the energy functional \cite{cances}] \label{lem:bounds}
For all $\gamma \in K$, where $K$ is the convex set defined in \eqref{eq:definition_set_k}, we get $\sqrt{\rho_\gamma} \in H^1(\R^3)$ and the following inequalities:
\begin{enumerate}[(i)]
\item Lower bound on the kinetic energy:
\begin{equation} \label{Lower bound on the kinetic energy}
\frac{1}{2} || \nabla \sqrt{\rho_\gamma} ||^{2}_{L^2} \leq   \tr[- \Lap \gamma]
\end{equation}

\item Upper bound on the Coulomb energy:
\begin{equation} \label{Bound-Coulomb}
0 \leq J[\rho_\gamma] \leq C \, \tr[\gamma]^{\frac{3}{2}} \tr[- \Lap \gamma]^{\frac{1}{2}}
\end{equation}

\item Bounds on the interaction energy between nuclei and electrons: \noeqref{interactionbounds}
\begin{equation}
- 4 Z \tr[\gamma]^{\frac{1}{2}} \tr[- \Lap \gamma]^{\frac{1}{2}} \leq \int_{\R^3} \rho_{\gamma}(x) V(x) \di x \leq 0 \label{interactionbounds}
\end{equation}

\item Bounds on the exchange-correlation energy:
\begin{equation} \label{Bounds on the exchange-correlation energy}
-C \left( \tr[\gamma]^{1 - \frac{\beta_{-}}{2}}
\tr[- \Lap \gamma]^{\frac{3 \beta_{-}}{2}}  
+ \tr[\gamma]^{1 - \frac{\beta_{+}}{2}} \tr[- \Lap \gamma]^{\frac{3 \beta_{+}}{2}} \right) 
\leq E_{xc}[\rho_\gamma] \leq 0
\end{equation}

\item Lower bound on the energy: \noeqref{energybound}
\begin{equation}
\E[\gamma] 
\geq \frac{1}{2} \left( \tr[-\Lap \gamma]^{\frac{1}{2}} - 4 Z\tr[\gamma]^{\frac{1}{2}} \right)^{2} - 8 Z^2 \tr[\gamma] - C \left( \tr[\gamma]^{\frac{2- \beta_{-}}{2 - 3 \beta_{-}}} + \tr[\gamma]^{\frac{2- \beta_{+}}{2 - 3 \beta_{+}}} \right) \label{energybound}
\end{equation}

\item Lower bound on the energy at infinity: \noeqref{energyinfinity}
\begin{equation}
\E^{\infty}[\gamma] \geq \frac{1}{2} \tr[-\Lap \gamma] - C \left( \tr[\gamma]^{\frac{2- \beta_{-}}{2 - 3 \beta_{-}}} + \tr[\gamma]^{\frac{2- \beta_{+}}{2 - 3 \beta_{+}}} \right). \label{energyinfinity}
\end{equation}
\end{enumerate}
In particular, minimizing sequences of $I_\lambda$ \eqref{eq:ground_state_molecule} and  $I^{\infty}_\lambda$ \eqref{eq:problem_at_infinity} are bounded in $\H$.

\end{lemma}{}

Lemma \ref{lem:bounds} is a central point for the existence of minimizer in the fixed nuclei setting but more importantly for us it bounds the minimizers independently of the position of the nuclei.

Let us now restate the main result of this paper.
\begin{theorem}[Dissociation limit] \label{thm:main}
Let $I^{X_2}_{\lambda, R}$ and $I^{X}_\lambda$ be defined by \eqref{eq:ground_state_molecule}, then we have for positively and neutrally charged molecules, i.e.~for $\lambda \leq 2Z$,
\begin{equation}\label{eq:main}
    \lim\limits_{R \to \infty} I_{\lambda, R}^{X_2} = \min \limits_{\alpha \in [0,\lambda]} \big( I^{X}_\alpha + I^{X}_{\lambda - \alpha} \big). 
\end{equation}{}
\end{theorem}
\begin{proof}
The proof of Theorem \ref{thm:main} is quite technical and is split into several parts. Since we first want to concentrate on its implications and in order to help with the reading flow of this paper we moved it into its own Section \ref{sec:proof}.
\end{proof}
Theorem \ref{thm:main} says the energy of the $X_2$-molecule converges -- as the nuclei are pulled infinitely far apart -- to the minimum over distributing the amounts of electrons $\lambda$ on two separated $X$-atoms.
For linear problems this directly gives $2I^X_{\nicefrac{\lambda}{2}}$, i.e.~a symmetric splitting, but for nonlinear problems $\alpha \mapsto I^X_\alpha + I^X_{\lambda-\alpha}$ might take its minimum at another value.
Whether the right hand side gives the expected symmetric minimum or not, will be discussed on the basis of the $H_2$ molecule in the next Section.

We want to stress again that we consider the spin-restricted setting also for the two individual atoms. Hence applying it to an $H$-Atom with a single electron has to be taken with a grain of salt.

Before we move on, let us make the following remarks regarding the modeling perspective of Theorem \ref{thm:main}.

\begin{remark}[possible generalizations]
We remark that the result of Theorem \ref{thm:main} can be extended to a large class of finite systems, i.e. arbitrary molecules, provided we stay in a neutral or positively charged setting and the minimum distance between the nuclei of each sub-system goes to zero. 

To ease presentation and focus on the main implications, we concentrate on the diatomic case, most prevalent in the related physics literature (see e.g. \cite{lu_symmetrybreaking,jackson1987dynamics,Zhang2004DissociationEO}).
Moreover Theorem \ref{thm:main} does not guarantee a symmetric splitting even in the case of hydrogen, so a simple extension to larger systems seems to be of limited interest. 

\end{remark}

\begin{remark}[spin-restricted vs. spin-unrestricted]
For the question of existence of minimizers the extension to spin-polarized systems was of independent interest \cite{cances, Gontier_2014}. 
Conversely, we expect that the equivalent of Theorem \ref{thm:main} would be a quantitative result about distribution of mass of the eigenfunctions of the density matrix, which does not favor a straightforward interpretation in terms of individual orbitals or electrons.
This also seems to be the reason why the spin-restricted case is more prevalent in practical computations.

Finally, staying in the spin-restricted setting allows for a comparison with models outside the Kohn-Sham framework as done in Subsections \ref{sec:one_dim} and \ref{sec:three_dim}.

So unless new techniques are developed to improve on our result an extension in this direction also seems to be of limited interest.




\end{remark}


\section{Symmetric Dissociation or not ? \label{sec:dissociation}}

The question which arises now is of course: Does 
    \begin{align} \label{eq:minimum}
           \min \limits_{\alpha \in [0,1]} \big( I^{H}_\alpha + I^{H}_{2 - \alpha} \big) \overset{?}{=} 2 I^H_1,
    \end{align}

hold or not, i.e.~do we have the right dissociation limit which we expect from physical intuition or which holds also for the Schr\"odinger equation.
The answer is it depends on the ``strength'' of the exchange-correlation functional.
To discuss this further we consider in the following only the Dirac exchange $e_{xc}(\rho) = - c_{xc} \rho^{\nicefrac{4}{3}}$ -- the prototypical example arising from the homogeneous electron gas -- with the constant $c_{xc}$ determining the strength of the exchange term.

To get a better feeling for what determines if the splitting is symmetric or not, i.e. if $\alpha =1$ is the minimizer in \eqref{eq:minimum}, we consider first a one-dimensional model.
\subsection{One-dimensional model \label{sec:one_dim}}
As we will see in the following section, it is quite hard to determine when 
\begin{align*}
    \min_{\alpha \in [0,1]} \big( I^{H}_\alpha + I^{H}_{2-\alpha} \big) = 2 I_{1}^{H}.
\end{align*}{}

To understand the problem better we study in this section the one-dimensional problem.
Since the Coulomb potential is not well suited for the one dimensional case, we consider $v(x) = \delta_0(x)$, i.e.~a simple contact potential \cite{dft_1d_contact}. 
The corresponding full Schr\"odinger system for the $H_2$-molecule looks like

\begin{align} \label{eq:h2_hamiltonian_one_dim}
    I_{R}^{H_2} = \inf \limits_{\substack{\psi \in H^1(\R_\Sigma^2),\\ \norm{\psi}_{L^2}=1,\\ \psi \text{ antisymm.}}} \langle \psi , H(x,y) \psi \rangle, \quad 
H(x,y) = \sum_{z \in \{x,y\}} -\frac{1}{2} \diff[2]{}{z} - \delta_0(z)  -\delta_R(z) + \delta_{|x-y|}(z) 
\end{align}

and the energy for the $H$-atom becomes
\begin{align}
    I^{H} = \inf \limits_{\substack{\phi \in H^1(\R),\\ \norm{\phi}_{L^2}=1}} \langle \phi , h(x) \phi \rangle, \quad 
h(x) =  -\frac{1}{2} \diff[2]{}{x} - \delta_0(x) . 
\end{align}
Note that $\delta_{|x-y|}(x)$ in \eqref{eq:h2_hamiltonian_one_dim} denotes the delta-distribution, i.e.
\[
\int_\R \delta_{|x-y|}(x) f(x,z) \di x = f(y,z)
\]

As for the standard Schr\"odinger system also here we have the right dissociation limit.

\begin{proposition}[Dissociation limit for the Schr\"odinger setting] \label{prop:linear_case}
For the full Schr\"odinger setting we always have
\begin{align}
    \limit{R}{\infty} I^{H_2}_{R} = 2 I^H,
\end{align}
i.e. the right dissociation limit.
\end{proposition}{}
\begin{proof}
See Section \ref{sec:proof_linear_case}.
\end{proof}{}
Note that Theorem \ref{thm:main} gives exactly the same result in the nonlinear case, but in the linear case every pair $(\alpha, 2-\alpha)$ gives the same result, so we always have symmetric dissociation.


Now we consider the DFT version of this system.
Note that in this case the Hartree term takes the form
\begin{align*}
    J[\rho] = \frac{1}{2} \int \int \rho(x) v(x-y) \rho(y) \di x \di y
    = 
    \frac{1}{2} \int \rho^2(x) \di x.
\end{align*}{}
Furthermore the exchange energy per volume looks like $e_{xc}(\rho) = - c_{xc} \rho^2$, where the exponent is $2 = 1 + \tfrac{1}{d}$ and $c_{xc} = \tfrac{1}{4}$ see \cite{dft_1d_contact, laestadius2019onedimensional}.

In total our energy functional for the $H$-atom takes the form 
\begin{align}\nonumber 
    \E^{H}[\rho]
    &=
    \frac{1}{2} \int \big( \sqrt{\rho}'\big)^2 \di x
    - \int v \rho \di x 
    + \frac{1}{2} \int \int \rho(x) \rho(y) v(|x-y|) \di x \di y
    + E_{xc}[\rho] \\ \label{eq:one-dim_dft}
    &=
    \frac{1}{2} \int \big( \sqrt{\rho}'\big)^2 \di x
    - \rho(0)
    + \big(\tfrac{1}{2} - c_{xc}\big) \int \rho^2 \di x.
\end{align}
And analogously for the $H_2$-molecule
\begin{align*}
    \E^{H_2}[\rho]
    =
    \int \big( \sqrt{\rho}'\big)^2 \di x
    - \rho(0)
    -\rho(R)
    + \big(\tfrac{1}{2} - c_{xc}\big) \int |\rho|^2\di x.
\end{align*}{}

In the same way as in Section \ref{sec:proof} we can show the dissociation limit
\begin{align*}
    \limit{R}{\infty}I_R^{H_2}
    =
    \min\limits_{\alpha \in [0,1]} \bigg( I^{H}_\alpha + I_{2 - \alpha}^{H} \bigg).
\end{align*}{}

Due to replacing the Coulomb potential by a contact potential we simplify the problem because the Hartree and the exchange energy take the same form.
Hence the energy functional $\rho \mapsto \E^H[\rho]$ is clearly convex for $c_{xc} \leq \tfrac{1}{2}$, since the von-Weizsäcker kinetic energy is.

This property is inherited by the infimum. Take any $\rho_\alpha$, $\rho_\beta$ non-negative and with $L^1$-norm $\alpha$, $\beta$, respectively. Then, 
\begin{align*}
    I^H_{\lambda \alpha + (1-\lambda)\beta}
    \leq 
    \E^H[\lambda \rho_\alpha + (1-\lambda) \rho_\beta]
    \leq 
    \lambda \E^H[\rho_\alpha] + (1-\lambda) \E^H[\rho_\beta],
\end{align*}{}
taking the infimum over $\rho_\alpha$, $\rho_\beta$ gives the convexity of $\alpha \mapsto I_\alpha$.

Therefore we have for $c_{xc} \leq \tfrac{1}{2}$
\begin{align*}
    2 I^H_1 
    =
    2 I^H_{\tfrac{1}{2} \alpha + \tfrac{1}{2}(2-\alpha)}
    \leq 
    I^H_{\alpha} + I^H_{2- \alpha},
\end{align*}{}
so symmetric splitting occurs.

For $c_{xc} > \tfrac{1}{2}$ there is no symmetric splitting anymore.
In order to see this, note that taking the test-functions $(1\pm \eta)\rho_1$ with $\rho_1$ the minimizer to $I^H_1$ yields
\begin{align*} 
    I^H_{1+\eta} + I^H_{1-\eta}
    \leq 
    2 I^H_1 + 2 \eta^2 \big(\tfrac{1}{2 } - c_{xc} \big) \int \rho_1^2 \di x
    < 2 I^H_1,
\end{align*}{}
i.e.~$2I^H_1$ is the strict global maximum.
Furthermore in this setting, i.e.~the one-dimensional DFT system with contact potential given by the energy functional \eqref{eq:one-dim_dft}, the ground state density can be found explicitly, see e.g.~\cite{witthaut_1d_exactsolution}:
\begin{align} \label{eq:1d_wavefunction}
   \rho = \alpha |\psi|^2, \quad \text{with} \quad \psi(x) = a \cdot \mathrm{sech}\big( b |x| + x_0\big),
\end{align}{}
where the parameters $a,b,x_0$ only depend on $\alpha$ and $c_{xc}$ and are given by
\[
x_0 = \mathrm{arctanh}\big(  \frac{1}{b}\big),\quad
a = \sqrt{\frac{b^2}{2(b-1)}}, \quad
b = 1 - \alpha \frac{1-2c_{xc}}{2}.
\]
With this we obtain 
\begin{align}\label{eq:exact_one_dim_energy}
I^{H}_\alpha + I_{2 - \alpha}^{H}
=
\frac{1}{12} (\alpha^2 (3 - 12 c_{xc}^2) + 6 \alpha ( 4 c_{xc}^2-1) -
   4 (1 + 2 c_{xc} + 4 c_{xc}^2)),
\end{align}
where the exact integrals are carried out in the appendix.

Equation \eqref{eq:exact_one_dim_energy} directly implies
\[
\min\limits_{\alpha \in [0,1]} \bigg(
I^{H}_\alpha + I_{2 - \alpha}^{H} \bigg)
= I^{H}_2.
\]
Therefore for $c_{xc}> \frac{1}{2}$, we always have both electrons bound at one nucleus.

The fact that the minimum is attained at an integer, is also something we observe numerically in the three-dimensional case.
From the view point of physics this make sense since we can not split an electron in half, but it is non obvious why this drops out of the mathematics.

\subsection{The three-dimensional case \label{sec:three_dim}}

Now we go back to the physically more interesting case of three dimensions.
Since we are just considering the $H$-atom the kinetic energy is the same as the von Weiz\"acker kinetic energy, i.e.
\begin{align} \label{eq:weizaecker}
    \E[\rho] = \frac{1}{2} \intr |\nabla \sqrt{\rho}|^2 \di x 
    + \intr V \rho \di x + J[\rho] + \intr e_{xc}(\rho) \di x
\end{align}{}
and with the energy as before 
\begin{align*}
 E_{\alpha} =   \inf_{\rho \in \mathcal{A}_\alpha} \E[\rho], 
 \quad \mathcal{A}_\alpha := \{ \rho \in L^1 : \sqrt{\rho} \in H^1(\R^3), \intr \rho \di x = \alpha \}.
\end{align*}{}
In this section we assume the exchange functional in \eqref{eq:weizaecker} is given by   $e_{xc}(\rho) = - c_{xc} \rho^{\nicefrac{4}{3}}$ (Dirac-exchange). 
\begin{remark}[TFDW]
In the remainder of this section we will be considering the Thomas-Fermi-Dirac-von Weizsäcker (TFDW) energy functional given by \eqref{eq:weizaecker}. As mentioned this coincides with our DFT energy functional in the case of the H-atom, i.e. $N=1$. So this case is the most interesting for our results. Nevertheless we will be considering an arbitrary real number $N$ of electrons in the following statements whenever possible. 
\end{remark}

Then for $c_{xc} \gg 1$ we observe symmetry breaking as in the one-dimensional case.
\begin{proposition}[Neutrally charged case] \label{thm:splitting}
For $c_{xc} =0$ we have the correct splitting, i.e.~$\alpha = 1 $ is the unique global minimizer to $\alpha \mapsto E_\alpha + E_{2-\alpha}$.
On the other hand there exists a $c(N) > 0 $ such that if $c_{xc} > c(N)$ we obtain 
    \[
    2 E_N  >  \big( E_\alpha + E_{2N - \alpha} \big) \quad \forall \alpha \neq N
    \]
i.e.~symmetry breaking occurs.    
\end{proposition}{}

\begin{proof}

We start with  the extreme case with $c_{xc}=0$, then the functional $\rho \mapsto \E[\rho]$ is strictly
convex and hence we obtain for any admissible densities $\rho_\alpha, \rho_{2N-\alpha}$ with mass $\alpha$ and $2N-\alpha$, respectively
    \begin{align}\label{eq:convex_tfdw_minimizer}
    2 E_N 
    \leq
        2 \E[\tfrac{1}{2} \rho_\alpha + \tfrac{1}{2} \rho_{2N- \alpha}]
    < 
        \E[\rho_\alpha] + \E[\rho_{2N-\alpha}].
    \end{align}
Taking now the infimum over $\rho_\alpha$ and $\rho_{2N - \alpha}$ gives
\[
2 E_N =   \min \limits_{\alpha \in [0,N]} \big( E_\alpha + E_{2N- \alpha} \big).
\]
So here the minimum is really attained at the symmetric splitting.
Furthermore we also have that $\alpha =N $ is always the strict global minimizer.
Indeed this can be seen by a case distinguishment:

\underline{
Case 1:}
 Assume minimizers exist  also for slightly negatively charged systems, i.e.~there is some $\epsilon>0$ such that minimizer exist for all $\alpha \in [0,N + \epsilon]$. Note we do not assume anything about uniqueness of minimizers just existence.
 Then we directly get a strong inequality $2E_N < E_\alpha + E_{2N-\alpha}$ by taking the corresponding minimizers in \eqref{eq:convex_tfdw_minimizer} for $\alpha \neq N$.
 Thus, in this case $\alpha = N$ would be a strict local minimum and by convexity the unique global minimum. 
 
 \underline{Case 2:}
If we do not have a minimizer for slightly negatively charged systems, then this can only happen if $\alpha \mapsto E_\alpha$ is not strictly decreasing anymore for $\alpha >N$. 
Otherwise we would have a strict subadditivity inequality because here the problem at infinity $E^\infty_\beta =0$ is trivial and the strict subadditivity condition (compare Lemma \ref{lem:subadditivity}) would be give us existence directly \cite{lions_cc, lions_cc2}.

Additionally due to convexity and the fact that $\alpha \mapsto E_\alpha$ is always non-increasing, we must have $E_\alpha = E_N$ for every $\alpha \in [N,2N]$. But in this case we have 
\begin{align*}
    E_\alpha + E_{2N -\alpha} 
    =
    E_\alpha + E_N > 2 E_N \qquad \forall~ \alpha \in [0,N).
\end{align*}{}
This is what happens in Thomas-Fermi theory, compare Figure \ref{fig:plot_energy_TF}.
Therefore also in the second case the symmetric splitting is the minimum if we set the exchange constant $c_{xc}=0$.

Now we consider the second statement of our proposition, i.e.~we take $c_{xc}$ to be large:\\
Let $\rho_N$ be a minimizer of $E_N$ (which is known to exist \cite{lieb1997thomas}) and $\eta \in (0,N)$. Then
\begin{align*}
    E_{(N+\eta)} +  E_{(N-\eta)} 
    & \leq 
    \E[(1+\tfrac{\eta}{N})\rho_N] + \E[(1-\tfrac{\eta}{N})\rho_N] \\
    &=
    2 \big( T[\rho_N] + V[\rho_N] \big) + \big((1+\tfrac{\eta}{N})^2 +(1-\tfrac{\eta}{N})^2 \big) J[\rho_N] + \big( (1+\tfrac{\eta}{N})^{\nicefrac{4}{3}} +(1-\tfrac{\eta}{N})^{\nicefrac{4}{3}} \big) E_{xc}[\rho_N] \\
    &=
    2 E_N + \tfrac{\eta^2}{N^2} \bigg( 2J[\rho_N] +\tfrac{4}{9} E_{xc}[\rho_N] \bigg) + o\big(\tfrac{\eta^2}{N^2}\big),
\end{align*}{}
where we used the Taylor-expansion for $(1\pm\eta)^{\nicefrac{4}{3}}$.
Now we can use Hardy-Littlewood-Sobolev and then H\"older interpolation to bound $J[\rho_N]$.
\begin{align*}
    2J[\rho_N] = \intr \intr \frac{\rho_N(x) \rho_N(y)}{|x-y|} \di x \di y
    \leq
    C_{HLS} ||\rho_N||^2_{L^\frac{6}{5}}
    \leq 
    C_{HLS} \norm{\rho_N}_{L^1}^{\nicefrac{2}{3}} \norm{\rho_N}_{L^{\nicefrac{4}{3}}}^{\nicefrac{4}{3}}.
\end{align*}{}
So we get using $\norm{\rho_N}_{L^1} =N$
    \[
        2J[\rho_N] +\tfrac{4}{9} E_{xc}[\rho_N]
        \leq 
        \bigg( C_{HLS}N^{\nicefrac{2}{3}} - \frac{4}{9} c_{xc} \bigg) \intr \rho_{1}^{\nicefrac{4}{3}} \di x < 0,
    \]
for $c_{xc} > \frac{9}{4} C_{HLS} N^{\nicefrac{2}{3}}$.
Putting in the numbers, i.e.~using the optimal $C_{HLS}$ given in \cite{lieb_hls} we see that 
    \[
        c_{xc} > \frac{9}{4} \sqrt{\pi} \frac{\Gamma(1)}{\Gamma(\tfrac{5}{2})} \bigg( \frac{\Gamma(3)}{\Gamma(\tfrac{3}{2})}\bigg)^{\nicefrac{2}{3}} N^{\nicefrac{2}{3}}
        \approx
        5.1615~ N^{\nicefrac{2}{3}}
    \]
suffices. 
In this case the symmetric splitting of the mass is not the minimum, in fact it is the maximum since the remaining terms in the Taylor expansion all have negative sign. 

Note again that for our result of Theorem \ref{thm:main} concerning KS-DFT only the case $N=1$, i.e.~the H-atom, is covered by this argument.
\end{proof}{}

While Proposition \ref{thm:splitting} deals with the extreme cases $c_{xc } =0$ and $c_{xc} \gg 1$, we were not able to prove symmetric splitting for the physically most interest case $c_{xc} =  \tfrac{3}{4} \big(\tfrac{3}{\pi}\big)^{\nicefrac{1}{3}}$. 
Therefore, we studied the behavior numerically. As in the one-dimensional setting \ref{sec:one_dim}, the minimum seems to be always attained at an integer pair. But the transition from symmetric to asymmetric seems to be more interesting since the function $\lambda \mapsto I^{H}_{\lambda} + I^{H}_{2-\lambda}$ does not simply switch from convex to concave.

\begin{figure}[h!]
     \centering
     \begin{subfigure}[b]{0.32\textwidth}
         \centering
         \includegraphics[width=\textwidth]{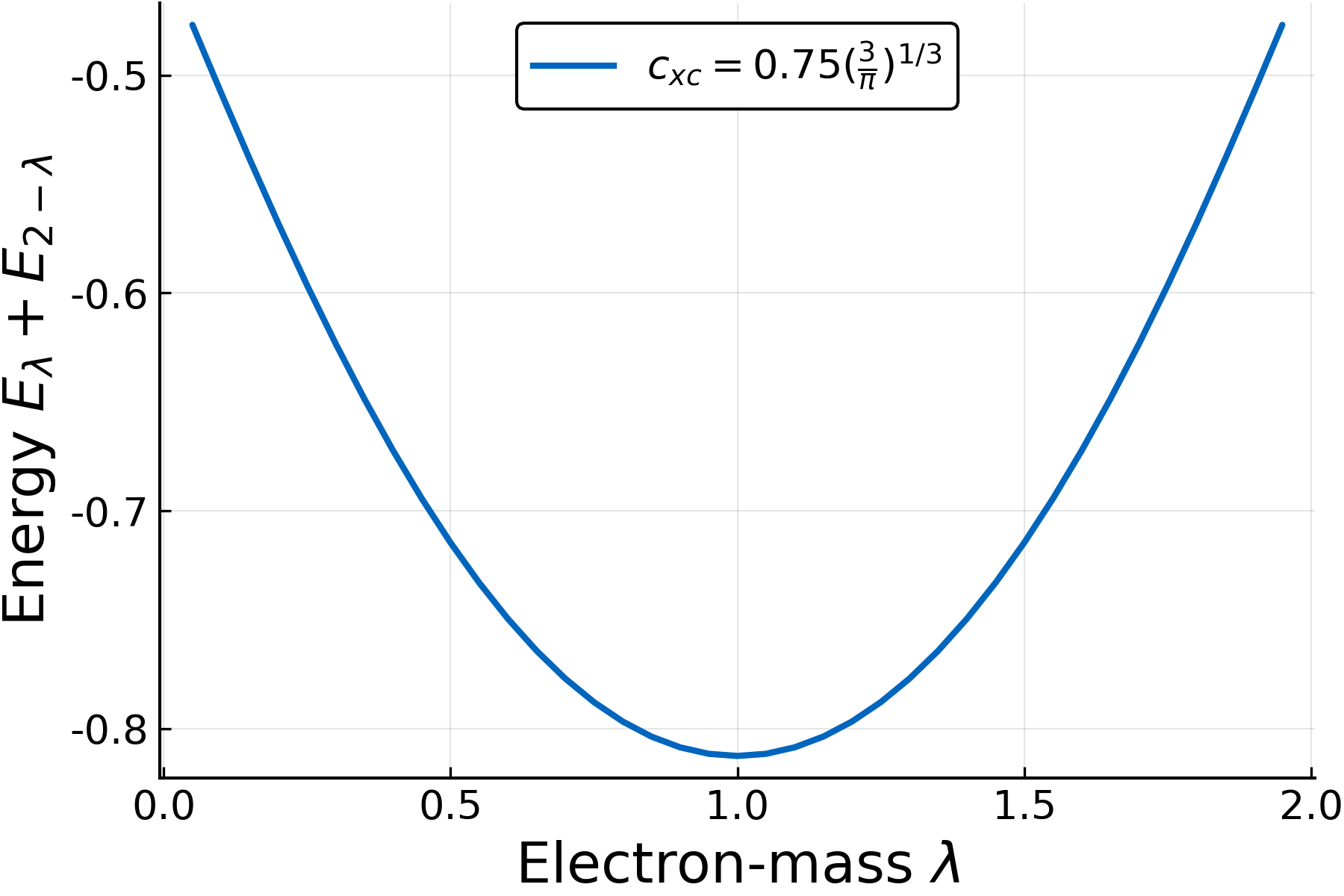}
     \end{subfigure}
     \hfill
     \begin{subfigure}[b]{0.32\textwidth}
         \centering
         \includegraphics[width=\textwidth]{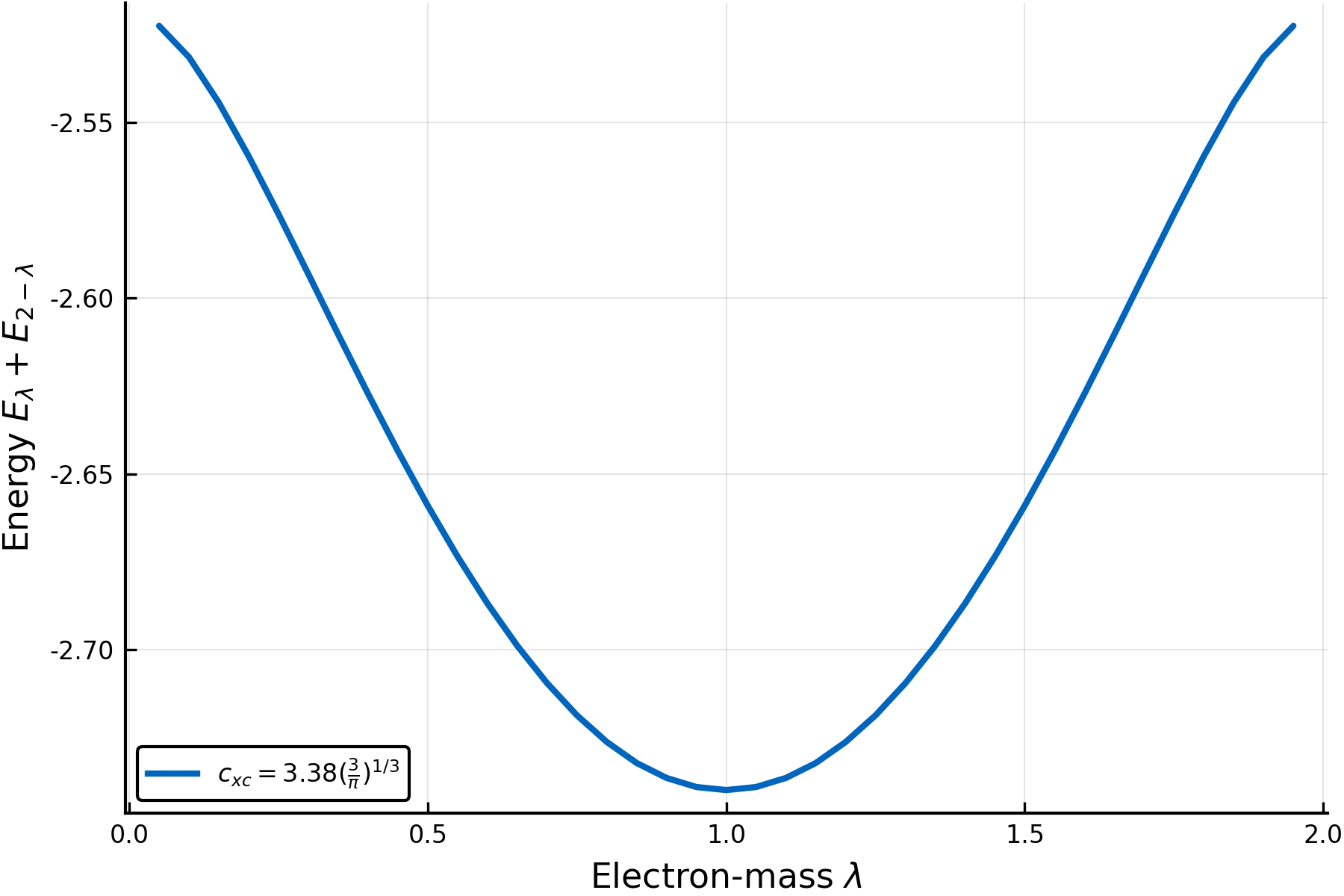}
     \end{subfigure}
     \hfill
     \begin{subfigure}[b]{0.32\textwidth}
         \centering
         \includegraphics[width=\textwidth]{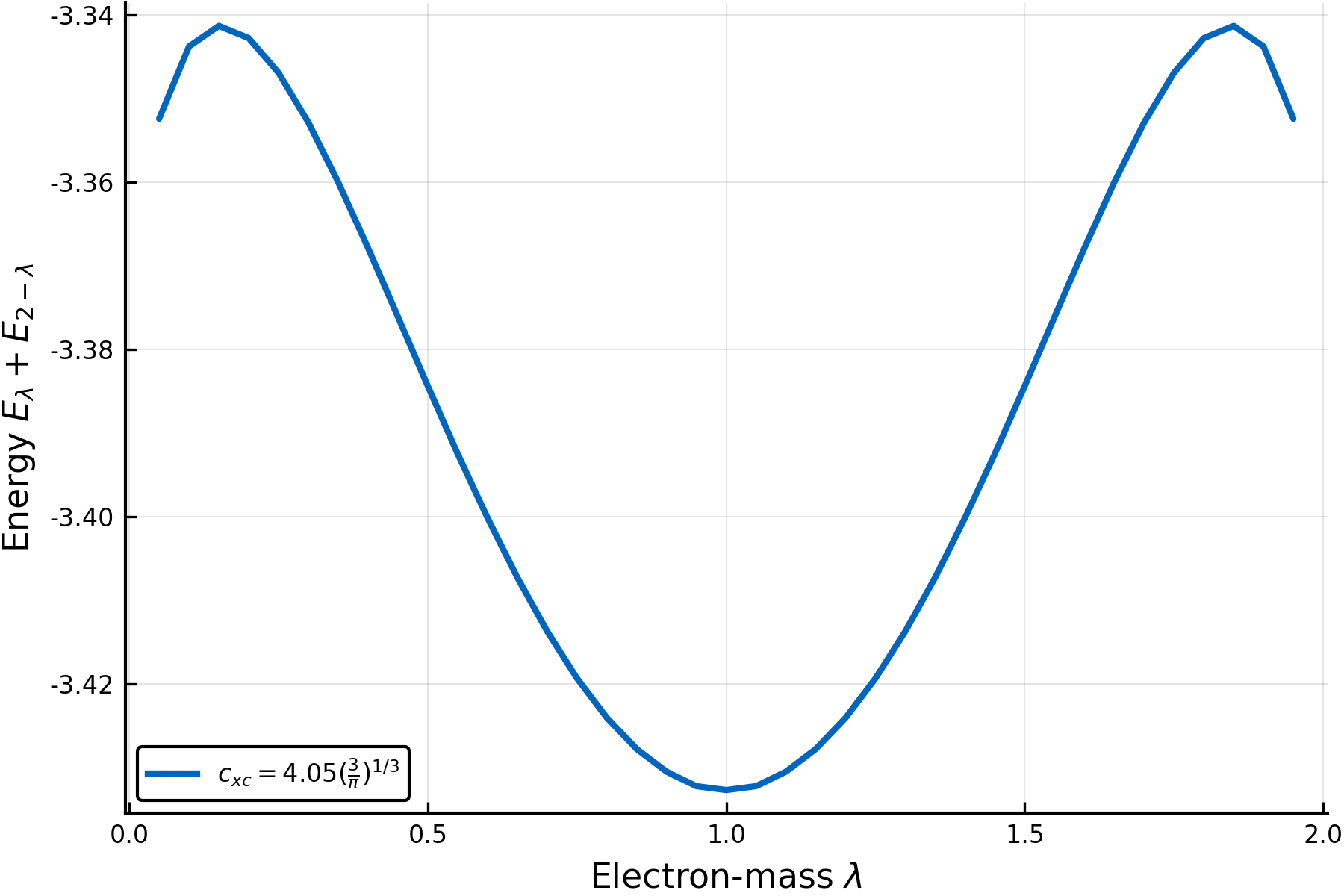}
     \end{subfigure}
     \vspace{5mm}
     
          \begin{subfigure}[b]{0.32\textwidth}
         \centering
         \includegraphics[width=\textwidth]{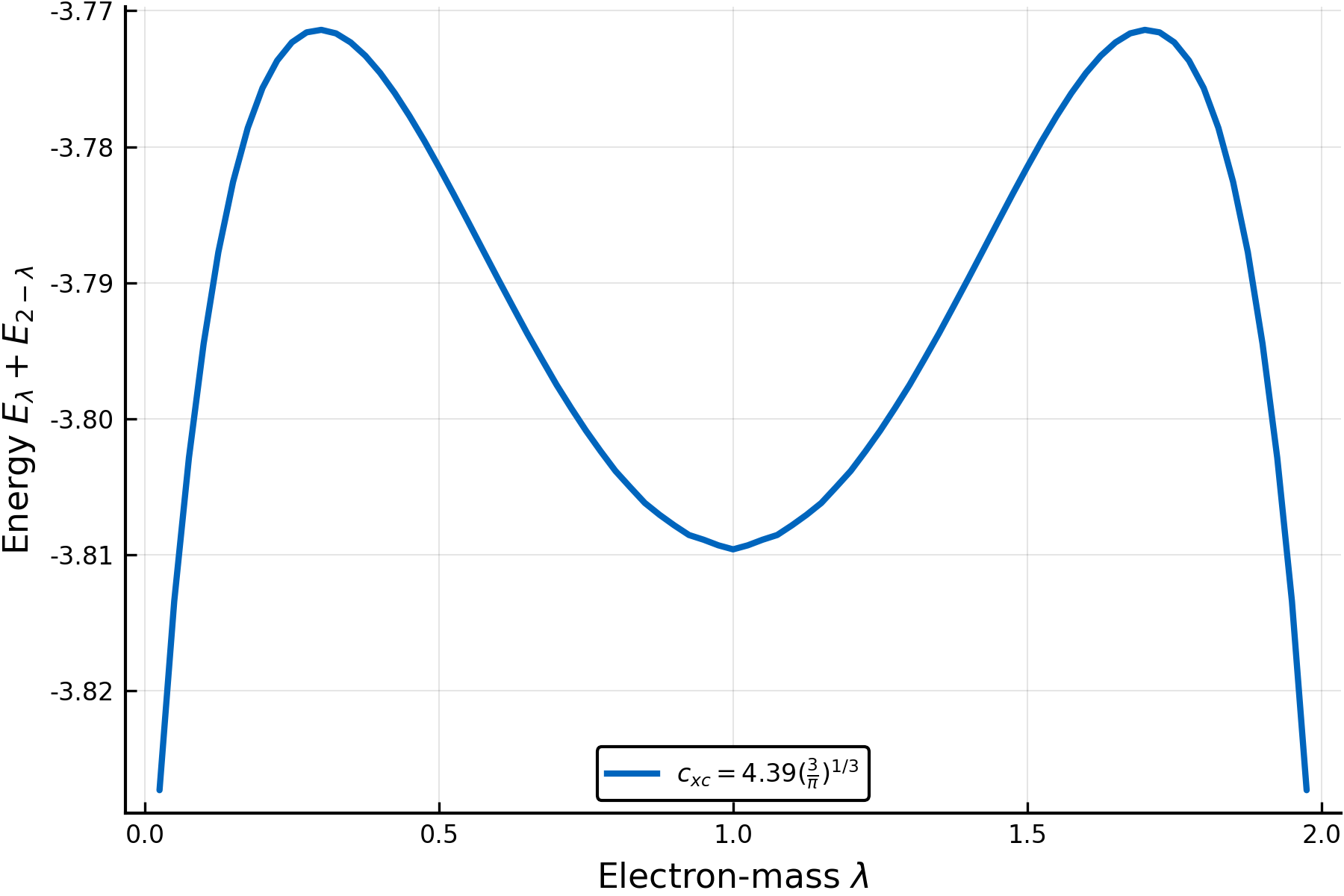}
     \end{subfigure}
     \hfill
     \begin{subfigure}[b]{0.32\textwidth}
         \centering
         \includegraphics[width=\textwidth]{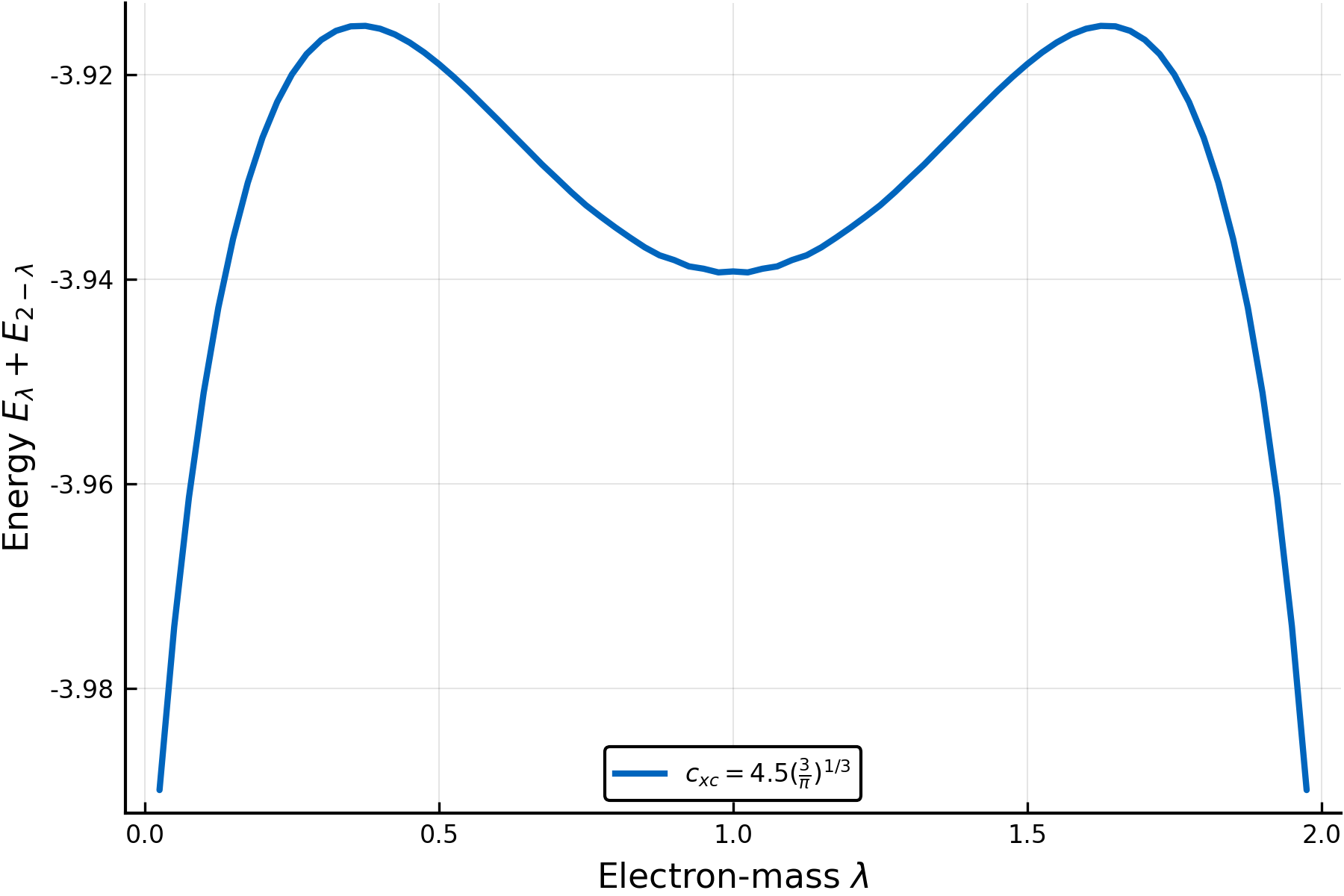}
     \end{subfigure}
     \hfill
     \begin{subfigure}[b]{0.32\textwidth}
         \centering
         \includegraphics[width=\textwidth]{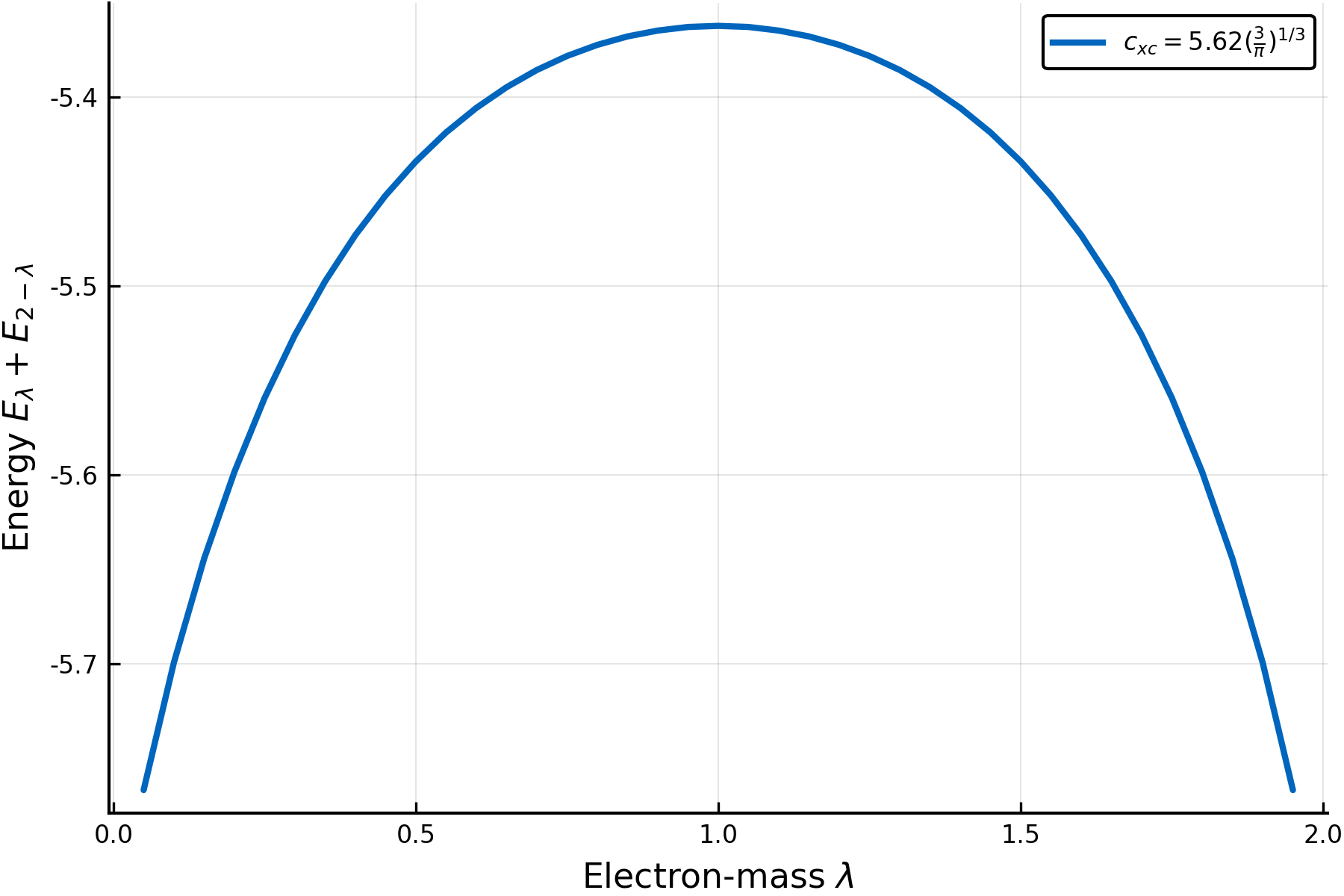}
     \end{subfigure}
     
        \caption{The function $\lambda \mapsto I^H_\lambda + I^{H}_{2-\lambda}$ for increasing values of $c_{xc}$.
        Note that the plot in the top left corner corresponds to the physically interesting case of $c_{xc} = \tfrac{3}{4} \big(\tfrac{3}{\pi}\big)^{\nicefrac{1}{3}}$; here we get numerically a symmetric splitting.}
        \label{fig:three graphs}
\end{figure}


The computations for Figure \ref{fig:three graphs} were done using the OCTOPUS package \cite{octopus}. We remark that after rescaling to an \(L^2\) normalized orbital, \(I^{H}_{\lambda}\) can be computed from a more standard DFT problem with modified electron-electron interaction potential (fractional charge) and modified exchange constant (cp. Slater X-$\alpha$ exchange \cite{slater_x_alpha}).

If we already start with a positively charged molecule in the beginning and thus wonder about the minimum of 
$\alpha \mapsto E_{2\lambda - \alpha} + E_{\alpha}$ for $\lambda < N$ we can get a stronger result.
\begin{proposition}[Positively charged case]
Let $\lambda < N$, then there exists a constant $c(\lambda) >0$ such that for all $c_{xc} < c(\lambda)$ we have 
\[
\min\limits_{\alpha \in [0,\lambda]} \big( E_{2\lambda - \alpha} + E_{\alpha} \big)= 2E_{\lambda}.
\]
\end{proposition}{}
\begin{proof}
As in the proof of Proposition \ref{thm:splitting} we know that for $c_{xc}=0$ the symmetric splitting is the strict global minimum.
By continuity it thus suffices to show that it stays a local one for all $c_{xc}$ small enough.

This follows by a result of Le Bris \cite{le1993some}. 
Indeed in Theorem 4 of \cite{le1993some} he proved that the mapping 
\(
\alpha \mapsto E_{\alpha}
\)
is strictly convex for $\alpha \leq Z $ and $c_{xc} >0$ small enough.
Note that Le Bris originally proved his convexity result for the Thomas-Fermi-Dirac-von Weizsäcker model. But since he considered an arbitrary non-negative constant in front of the Thomas-Fermi term,
this reduces to our model in the hydrogen case, if we set this constant to zero.

This directly implies
\[
2 E_{\lambda} = 2 E_{\frac{1}{2} (\lambda - \alpha) + \frac{1}{2} (\lambda + \alpha)} < E_{\lambda - \alpha} + E_{\lambda+ \alpha}, \quad \forall~\alpha \in (0,Z-\lambda).
\]
So the symmetric splitting is a local minimum, as mentioned above since it is a strict global minimum for $c_{xc} = 0$ and thus by continuity it remains a global minimum for $c_{xc} $ small enough.
\end{proof}{}

\section{Dissociation limit -- The proof \label{sec:proof}}
This section contains the proof to Theorem \ref{thm:main}, it is split into two parts containing the upper bound and lower bound, respectively.

\subsection{Upper bound}\label{sec:upper_bound}
We begin by proving the upper bound to Theorem \ref{thm:main}, i.e.
    \begin{equation} \label{eq:upper_bound}
        \limsup\limits_{R \to \infty} I_{\lambda, R}^{X_2} \leq \min \limits_{\alpha \in [0,\frac{\lambda}{2}]} \big( I^{X}_\alpha + I^{X}_{\lambda - \alpha} \big). 
      \end{equation}
For this purpose, given $\epsilon>0$ take $\gamma_\alpha \in K_\alpha$ and $\gamma_{\lambda - \alpha} \in K_{\lambda - \alpha}$, s.t.
    \[
        \E[\gamma_\alpha] \leq I^X_\alpha + \tfrac{\epsilon}{2} \quad \text{and} \quad\E[\gamma_{\lambda - \alpha}] \leq I^X_{\lambda - \alpha} + \tfrac{\epsilon}{2}.
    \]
Thanks to the continuity of the energy functionals established in Lemma \ref{lem:continuity_E} and the fact that the finite rank operator and
the functions $C^\infty_{c}(\R^3)$ are dense in $\H$ and $L^2(\R^3)$,  respectively, we may assume that both $\gamma_\alpha$ and $\gamma_{\lambda - \alpha}$ have finite rank with range in $C^{\infty}_{c}(\R^3)$.

Then define the operator $\gamma_R := \gamma_\alpha + \tau_{R} \gamma_{\lambda - \alpha} \tau_{-R}$, where $\tau_{R}$ is the unitary operator on $L^2(\R^3)$ defined by
\[
\big( \tau_{R} f\big) (x) := f(x-R).
\]
For $R$ large enough we have $\gamma_{R} \in K_\lambda$ and thus
\begin{align*}
    I^{X_2}_{\lambda, R} &\leq 
    \E^{X_2}_R [\gamma_R] \\
    &\leq
    \E^X[\gamma_\alpha] + \E^X[\gamma_{\lambda - \alpha}]
    + \intr \intr \frac{\rho_{\gamma_{\lambda - \alpha}} (x- R) \rho_{\gamma_{\alpha}}(y)}{|x-y|} \di x \di y \\
    &\leq I^X_\alpha + I^X_{\lambda - \alpha} + \epsilon + \intr \intr \frac{\rho_{\gamma_{\lambda - \alpha}} (x- R) \rho_{\gamma_{\alpha}}(y)}{|x-y|} \di x \di y \gegen{R}{\infty} I^X_\alpha + I^X_{\lambda - \alpha} + \epsilon ,
\end{align*}{}
where we used that $\rho_{\gamma_{\lambda - \alpha}}$ and $\rho_{\gamma_{ \alpha}}$ have compact support.
Taking the limsup yields
 \begin{align*}
      \limsup\limits_{R \to \infty} I_{\lambda, R}^{X_2} \leq
      I^X_\alpha + I^X_{\lambda - \alpha} + \epsilon.
 \end{align*}{}
Since $\epsilon >0$ and also $\alpha \in [0, \lambda]$ were arbitrary, we get the desired assertion.

\subsection{Lower bound}
The lower bound is more difficult, we want to prove 
\begin{align} \label{eq:lower_bound}
    \liminf\limits_{R \to \infty} I_{\lambda, R}^{X_2} \geq
      \min \limits_{\alpha \in [0,\lambda]} \big( I^{X}_\alpha + I^{X}_{\lambda - \alpha} \big).
\end{align}{}

Our proof idea is to use the concentration-compactness lemma, which is usually applied to a minimizing sequence, but this time act on a sequence of minimizers $\big(\gamma_{R_n}\big)_n$ of $\In$for a sequence $\big( R_n\big)_n$ tending to infinity.\\
 In the following we will denote the arising subsequence also with $\big(\gamma_{R_n}\big)_n$ to keep notation clearer.
Furthermore in the following $C>0$ will denote a generic constant, which may have different values
at each appearance, indicating some finite positive constant independent of the surrounding
variables.

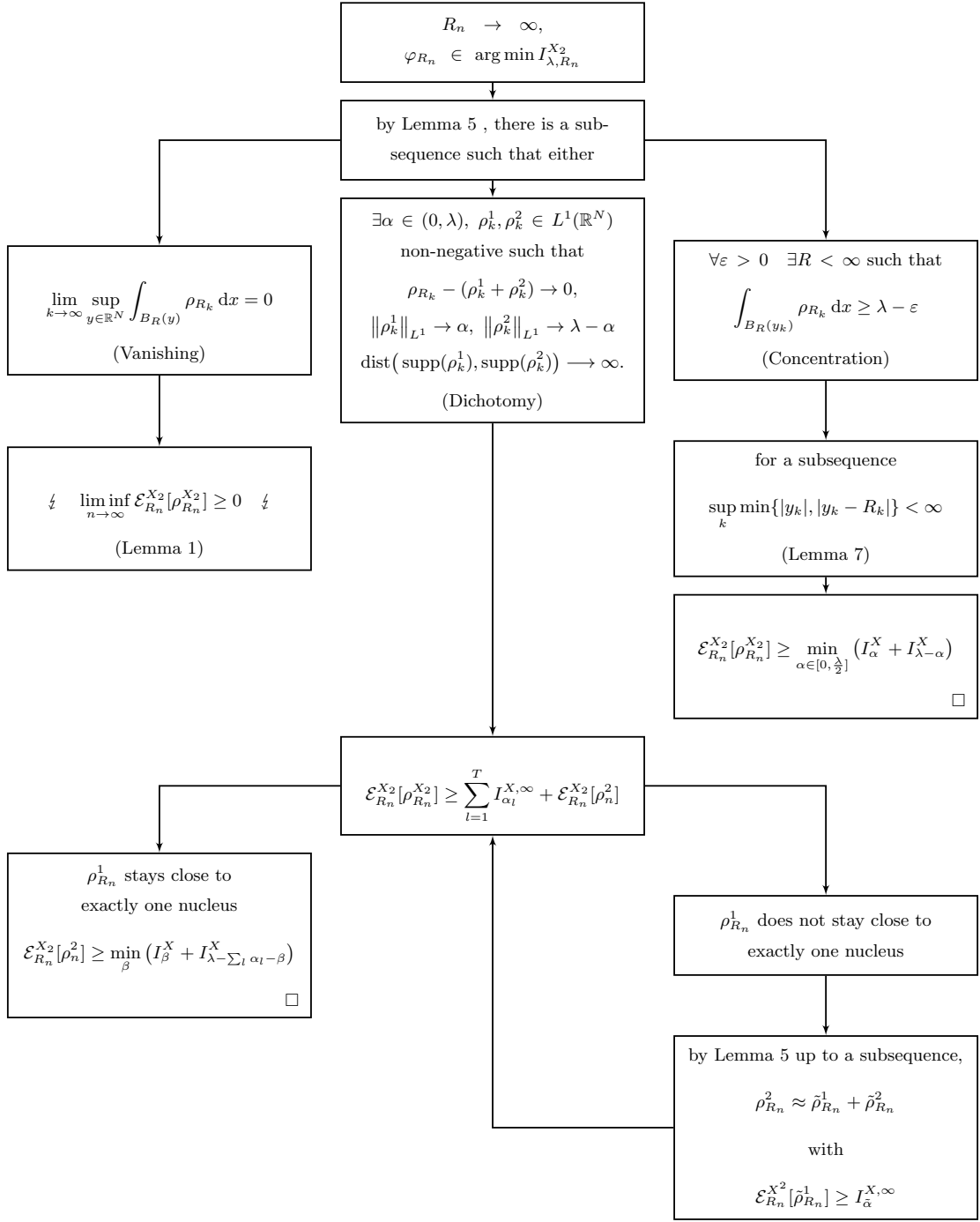
\begin{figure}[p]
  \begin{center}
  \noindent\adjustbox{width=\textwidth}{\begin{tikzpicture}[auto,
    decision/.style={diamond, draw=black, thick, fill=white,
    text width=8em, text badly centered,
    inner sep=1pt, font=\small},
    block_center/.style ={rectangle, draw=black, thick, fill=white,
      text width=14em, text centered,
      minimum height=4em, font=\small, inner sep=6pt},
    block_left/.style ={rectangle, draw=black, thick, fill=white,
      text width=16em, text ragged, minimum height=4em, inner sep=6pt, font=\small},
    block_noborder/.style ={rectangle, draw=none, thick, fill=none,
      text width=18em, text centered, minimum height=1em, font=\small},
    block_assign/.style ={rectangle, draw=black, thick, fill=white,
      text width=18em, text ragged, minimum height=3em, inner sep=6pt, font=\small},
    block_lost/.style ={rectangle, draw=black, thick, fill=white,
      text width=16em, text ragged, minimum height=3em, inner sep=6pt, font=\small},
      line/.style ={draw, thick, -latex', shorten >=0pt}]
    \matrix [column sep=5mm,row sep=3mm] {
      & \node [block_center] (start) {\(R_n \to \infty\), \\ \(\varphi_{R_n} \in \argmin \In \)}; & \\
      & \node [block_center] (ccp) {by Lemma \ref{lem:lions_cc} , there is a subsequence such that either}; & \\
      \node [block_center] (vanishing) {\[\lim \limits_{k \to \infty} \sup_{y \in \R^N} \int_{B_R(y)} \rho_{R_k} \di x =0 \] (Vanishing)}; & \node [block_center] (dichotomy) {\( \exists \alpha \in (0, \lambda), \ \rho_{k}^1, \rho_{k}^{2} \in L^1(\R^N)\) non-negative such that 
      \begin{gather*}
      \rho_{R_k} - (\rho^{1}_k + \rho^2_{k}) \to  0,\\
      \norm{\rho^{1}_k}_{L^1} \to \alpha, \ \norm{\rho^{2}_k}_{L^1} \to \lambda - \alpha \\
      \mathrm{dist}\big(\supp (\rho^{1}_k), \supp (\rho^{2}_k) \big) \longrightarrow \infty.
      \end{gather*} (Dichotomy)}; &\node [block_center] (concentration) {\(\forall \epsilon >0 ~~~ \exists R < \infty\) such that \[\int_{B_R(y_k)} \rho_{R_k} \di x \geq \lambda - \epsilon \] (Concentration)};\\
      \node [block_center] (vanishing_impossible) {\[\lightning \quad \liminfn \E^{X_2}_{R_n} [\rhon] \geq 0  \quad \lightning\] (Lemma \ref{lem:subadditivity})}; & &  \node [block_center] (concentration_stay_close) {for a subsequence \[ \sup_k \min \{\abs{y_k}, \abs{y_k - R_k}\} < \infty \] (Lemma \ref{lem:concentration_yk_bounded})}; \\
      & & \node [block_center] (concentration_done) {\[\E^{X_2}_{R_n}[\rhon] \geq \minimum\] \hfill $\square$}; \\
      & \node [block_center] (dichotomy_splitting) {\[ \E^{X_2}_{R_n}[\rhon] \geq \sum_{l=1}^T I^{X,\infty}_{\alpha_l} + \E^{X_2}_{R_n}[\rho_{n}^2]\]}; & \\
      \node [block_center] (two_masses) {\(\rho_{R_n}^1\) stays close to \\exactly one nucleus \[\E^{X_2}_{R_n}[\rho_{n}^2]\geq \min \limits_{\beta} \big( I^{X}_\beta + I^{X}_{\lambda - \sum_l \alpha_l - \beta} \big)\] \hfill $\square$}; & & \node [block_center] (further_splitting) {\(\rho_{R_n}^1\) does not stay close to \\exactly one nucleus}; \\
       & & \node [block_center] (energy_update) {by Lemma \ref{lem:lions_cc} up to a subsequence, \[\rho_{R_n}^2 \approx \tilde{\rho}_{R_n}^1 + \tilde{\rho}_{R_n}^2\] with \[\E^{X^2}_{R_n}[\tilde{\rho}^1_{R_n}] \geq I^{X,\infty}_{\tilde{\alpha}}\]};\\
    };
    \begin{scope}[every path/.style=line]
      \path (start) -- (ccp);
      \path (ccp) -| (vanishing);
      \path (ccp) -- (dichotomy);
      \path (ccp) -| (concentration);
      \path (vanishing) -- (vanishing_impossible);
      \path (concentration) -- (concentration_stay_close);
      \path (concentration_stay_close) -- (concentration_done);
      \path (dichotomy) -- (dichotomy_splitting);
      \path (dichotomy_splitting) -| (two_masses);
      \path (dichotomy_splitting) -| (further_splitting);
      \path (further_splitting) -- (energy_update);
      \path (energy_update) -| (dichotomy_splitting);
    \end{scope}
  \end{tikzpicture}}
  \caption{Structure of the proof for the lower bound. The loop on the bottom-right can only be visited a finite number of times.}
  \label{fig:structure_proof}
\end{center}
\end{figure}

\begin{lemma}[Lemma I.1., Lions \cite{lions_cc}] \label{lem:lions_cc}
Let $(\rho_n)_{n \geq 1} \subseteq L^1(\R^D)$ be a sequence of non-negative functions such that
$$ \int_{\R^D} \rho_n \di x = \lambda, \quad \lambda >0 \text{ fixed.}$$
Then there exists a subsequence $(\rho_{n_k})_{k \geq 1}$ satisfying one (and only one) of
the following properties:
\begin{enumerate}
\item \textcolor{blue}{Concentration} There is $(y_{n_k})_{k \geq 1} \subseteq \R^D$ such that $\rho(\cdot + y_{n_k})$ is tight, i.e. 
\[
\forall \epsilon >0 ~~~ \exists R < \infty: \quad \int_{B_R(y_k)} \rho_{n_k} \di x \geq \lambda - \epsilon  
\]
\item \textcolor{blue}{Vanishing} For any $R<\infty$, we have
\[ \lim \limits_{k \to \infty} \sup_{y \in \R^D} \int_{B_R(y)} \rho_{n_k} \di x =0, \]
\item \textcolor{blue}{Dichotomy}
  There is $\alpha \in (0, \lambda)$ and  $\rho_{k}^1, \rho_{k}^{2}$ $\in L^1(\R^D)$ non-negative such that 
  \begin{gather*}
  \int_{\R^D} |\rho_{n_k} - \rho^{1}_k - \rho^2_{k} | \longrightarrow  0,\\
  \int_{\R^D} \rho^{1}_k \longrightarrow \alpha  \text{ and } \int_{\R^D} \rho^{2}_k \longrightarrow \lambda - \alpha \\
  \mathrm{dist}\big(\supp (\rho^{1}_k), \supp (\rho^{2}_k) \big) \longrightarrow \infty.
  \end{gather*}
\end{enumerate}
\end{lemma}{}
For the dichotomy case we will actually use the stronger statement given in \cite{lions_cc2} see below.

Hence we have to distinguish three cases.
Note that the concentration case is the extreme case where the entire mass stays at one nucleus. Dichotomy corresponds to the electron mass being distributed in some way over the two nuclei. Finally vanishing means that the electron mass separates from both nuclei completely, which does not fit into our result.

Therefore let us start with the vanishing case.\\
\underline{Case 1: \textcolor{blue}{Vanishing}:}

We apply the bounds for the energy functional $\E^{X_2}$ established in \eqref{lem:bounds}, which yields
\begin{align}\label{eq:seminorm_bdd}
    \norm{\nabla \sqrt{ \rho^{X_2}_{R_n}}} _{L^2}^2
    \leq C + \E^{X_2}[\rhon] 
    =
    C + \In
    \leq C,
\end{align}{}
which implies that the sequence $\big( \rhon \big)_n$ is bounded in $H^1(\R^3)$.
Hence we can apply the following lemma by Lions.
\begin{lemma}[Lemma I.1., Lions\cite{lions_cc2}]
Let $1 \leq p \leq \infty$, $1 \leq q < \infty$ with $q \neq \tfrac{Dp}{D-p} =:p^{*}$ if $p <D$.
Assume that $(u_n)_{n \geq 1}$ and  $\big(\nabla u_n \big)_{n}$ are bounded in $L^q(\R^D)$ and $L^p(\R^D) $, respectively. If 
$$\sup_{y \in \R^D} \int_{B_R(y)} |u_n|^{q} \di x \gegen{n}{\infty} 0, \quad \text{for some } R >0,$$
then $u_n \to  0 $ in $L^\alpha(\R^D)$ for $\alpha$ between $q$ and $p^*$ (if $p \geq D$ set $p^{*} = \infty$).
\end{lemma}
With $p=q=2$ we obtain that 
\begin{align*}
    \sqrt{\rhon} \gegen{n}{\infty} 0, \quad \text{in} \enskip L^\alpha(\R^3), \enskip \alpha \in (2,6).
\end{align*}{}
So clearly we get by applying 1 and 3 from Assumption\ref{assumption1}
\begin{align*}
    0 \leq - E_{xc}[\rhon] = C\intr \big( \rhon\big)^{1+ \beta_+}  +\big( \rhon\big)^{1+ \beta_-} \di x \gegen{n}{\infty} 0,
\end{align*}
where $1+ \beta_{\pm} \in (1,\tfrac{5}{3})$ by assumption.\\
Furthermore we can split the Coulomb potential $V = v_1 + v_2$ with $v_1 \in L^{q}(\R^3)$ and $v_2\in L^r(\R^3)$ with $q<3$ and $r>3$. Hence by taking e.g.~$q=2, r=4$ and applying H\"older inequality with we obtain
\begin{align*}
    \bigg|\intr \frac{1}{|x|} \rhon \di x \bigg|
    \leq 
    \norm{v_1}_{L^2} \norm{\rhon}_{L^{2}} + \norm{v_2}_{L^4} \norm{\rhon}_{L^{\nicefrac{4}{3}}}
    \gegen{n}{\infty} 0. 
\end{align*}{}
Analogously for the second nucleus with Coulomb potential $\tfrac{1}{|\cdot - R_n|}$.
So combining those two results yields
\begin{align*}
    \liminfn \E^{X_2, R_n} [\gamman] \geq 0,
\end{align*}{}
but this contradicts the upper bound \eqref{eq:upper_bound} we established
\begin{align*}
    0 \leq \liminfn \E^{X_2, R_n} [\rhon] \leq \limsupn  \E^{X_2, R_n} [\rhon] \leq \minimum < 0 \enskip \lightning.
\end{align*}{}{}
Therefore vanishing cannot occur.

\underline{Case 2: \textcolor{blue}{Concentration}:}

Assume concentration occurs, i.e.~
\begin{align*}
    \forall~\epsilon>0 ~ \exists~(y_n)_n,~ M< \infty:~ \int_{B_M(y_n)} \rhon(x) \di x\geq \lambda - \epsilon \enskip \forall~n.
\end{align*}

Intuitively this corresponds to
\begin{align*}
    \E[\rhon] \gegen{n}{\infty} I^X_\lambda + I^X_0 = I^X_\lambda.
\end{align*}{}

We start off with a small lemma.

\begin{lemma}\label{lem:concentration_yk_bounded}
The sequence $(y_n)_n$ stays bounded around 0 or $(R_n)_n$, to be more precise (up to a subsequence)
\begin{align*}
    \exists~L<  \infty~~\forall n \geq0: \enskip |y_n| \leq L \text{ or } |y_n -R_n| \leq L.
\end{align*}{}
\end{lemma}{}

\begin{proof}
Assume $|y_n| > L $ and $|y_n-R_n| > L$ for any $L >0$ (in particular $L\gg M$).
We estimate the Coulomb interaction by applying Cauchy-Schwarz and then Hardy's inequality  to obtain
\begin{align*}
    \int \frac{\rhon}{|x|} \di x
    &= 
    \int \frac{\sqrt{\rhon} \sqrt{\rhon}}{|x|} \di x
    \leq 
    \bigg( \int \rhon \di x \bigg)^{\nicefrac{1}{2}} 
        \bigg( \int \frac{\rhon}{|x|^2} \di x \bigg)^{\nicefrac{1}{2}} \\
    &\leq 
    2 \bigg( \int \rhon \di x \bigg)^{\nicefrac{1}{2}} 
    \bigg( \int |\nabla \sqrt{\rhon}|^2  \di x \bigg)^{\nicefrac{1}{2}} 
    \leq C \bigg( \int \rhon \di x \bigg)^{\nicefrac{1}{2}}, 
\end{align*}{}
where we used that the $H^1$-seminorm of $\big(\sqrt{\rhon}\big)_n$ stays bounded \eqref{eq:seminorm_bdd}.
Now if $|y_n| >L$ we obtain 
\begin{align*}
    \bigg( \int_{B_{L-M}(0)} \frac{\rhon}{|x|} \di x \bigg)^2
    \leq 
    C \int_{B_{L-M}(0)}\rhon \di x 
    \leq 
    C\bigg( \lambda - \int_{B_M(y_n)} \rhon  \di x \bigg)
    \leq C \epsilon
\end{align*}{}
and the analogous result for the Coulomb interaction with the other nucleus.\\
Then,
\begin{align*}
    - \frac{1}{Z} V^{X_2}_{R_n}[\rhon] 
    &= 
    \intr \rhon \bigg( \frac{1}{|x|} + \frac{1}{|x-R|} \di x \bigg)\\
    &=
    \int_{B_{L-M}(0)} \frac{\rhon}{|x|} \di x + \int_{B_{L-M}(R_n)} \frac{\rhon}{|x-R_n|} \di x
    + \int_{B^c_{L-M}(0)} \frac{\rhon}{|x|} \di x + \int_{B^c_{L-M}(R_n)} \frac{\rhon}{|x-R_n|} \di x \\
    &\leq 
    \frac{2 \lambda}{L-M} + 2C \epsilon^{\nicefrac{1}{2}}.
 \end{align*}{}
 Since this inequality holds for any $L>M$ we can take $L \to \infty$ and then $\epsilon \to 0$, which gives
 \begin{align*}
     V^{X_2}_{R_n}[\rhon] \gegen{n}{\infty} 0.
 \end{align*}{}
 But this would imply
 \begin{align}\label{eq:contradiction_concentration}
     I^X_\lambda \geq \minimum
     \geq \limsupn \E^{X_2}_{R_n}[\rhon]
     \geq \liminfn \E^{X_2}_{R_n}[\rhon]
     = \liminfn \E^{\infty}[\rhon]
     \geq I^{\infty}_\lambda,
 \end{align}{}
 where the first inequality simply comes from the fact that $I^X_\lambda$ is the same as setting $\alpha = 0$, thus it is for sure bigger than the minimum. The second inequality is exactly our upper bound \eqref{eq:upper_bound} proven in the previous subsection.
But equation \eqref{eq:contradiction_concentration} is a contradiction to the strict inequality in Lemma \ref{lem:subadditivity} (ii).
 This finishes the proof.
\end{proof}{}
So Lemma \ref{lem:concentration_yk_bounded} gives us either $|y_n| \leq L$ or $|y_n - R| \leq L$ for some $L>0$.
W.l.o.g.~we can in the following assume that $|y_n| \leq L$ (otherwise transform the coordinate system by a reflection s.t. $0$ gets mapped to $R_n$. This leaves the energy functional unchanged.)\\
The last step consists now in a cut-off argument.

By Lemma \ref{lem:bounds} the sequence $(\gamman)_{n}$ stays uniformly bounded in $\H$ and hence we have (up to subsequence) 
\[
\gamman \overset{\ast}{\rightharpoonup} \gamma^{\ast} \text{ in } \H,\quad \sqrt{\rhon} \rightharpoonup \sqrt{\rho^\ast} \text{ in } H^1(\R^3).
\]
Since we are in the concentration case and the $(y_n)_n$ stays bounded, we can choose for any $\epsilon >0$ a compact set $Q \subseteq \R^3$ such that
\[
\int_Q \rhon \di x \geq \int_{B_M(y_n)} \rhon \di x \geq \lambda - \epsilon.
\]
Therefore we get for the limit $\rho^\ast$ using convergence in $L^1_{loc}$ 
\[
 \norm{\rho^\ast}_{L^1} \geq \int_Q \rho^{\ast} \di x = \limit{n}{\infty} \int_Q \rhon \di x \geq \lambda - \epsilon,
\]
since $\epsilon >0$ was arbitrary we get $\norm{\rho^\ast}_{L^1} = \lambda$.
Therefore $\sqrt{\rhon}$ convergences also strongly in $L^2$ and due to the weak convergence in $H^1$ also strongly in $L^p(\R^3)$ for $p \in [2,6)$.

Therefore we get by using the same line of argument as above with Hardy's inequality
\begin{align*}
    \intr \frac{1}{|x-R_n|} \rhon \di x \gegen{n}{\infty} 0.
\end{align*}{}
Using now the sequential weak lower semi-continuity of the kinetic energy functional $T$ we obtain
\begin{align*}
    \liminfn \E^{X_2}[\gamman] 
    \geq 
    \E^{\infty}[\gamma^{\ast}] - \intr \frac{1}{|x|} \rho^{\ast} \di x = 
    \E^{X}[\gamma^{\ast}] 
    \geq I^{X}_{\lambda} = I^{X}_{\lambda} + I^{X}_0.
\end{align*}{}

\underline{Case 3: \textcolor{blue}{Dichotomy}:}

Take a smooth partition of unity $\xi^2 + \zeta^2 =1$ such that 
\begin{align*}
    0\leq \xi, \zeta \leq 1, ~ \xi(x) = 1, \text{ if } |x| \leq 1, \xi(x) = 0 \text{ if } |x| \geq 2 ~~ \text{and} ~ \zeta(x) = 0, \text{ for } |x| \leq 1, \zeta(x) = 1 \text{ for } |x| \geq 2.
\end{align*}
Furthermore assume 
\begin{align*}
    \norm{\nabla \xi}_\infty \leq 2 \text{ and }  \norm{\nabla \zeta}_\infty \leq 2,
\end{align*}{}
and consider the dilated functions $\xi_K(x) = \xi\big(\tfrac{x}{K}\big)$ and
$\zeta_K(x) = \zeta\big(\tfrac{x}{K}\big)$.
Now if we use the detailed construction of the dichotomy case given in \cite{lions_cc2} (compare also \cite{cances}), we can assume that (up to a subsequence), there exists 
\begin{itemize}
    \item  $\alpha \in (0, \lambda) $
    \item  a sequence of points $(y_n)_n \in \R^3$
    \item  two increasing sequences of positive real numbers $(K^{(1)}_{n})_n$ and $(K^{(2)}_{n})_n$ such that 
        \begin{align} \label{eq:dichotomy_increasing}
            \limit{n}{\infty} K^{(1)}_{n} = \infty ~~ \text{and} ~~ \limit{n}{\infty} \frac{K^{(2)}_{n}}{2} - K^{(1)}_{n} = \infty
        \end{align}
\end{itemize}
such that the sequences $\gamma^{(1)}_{n} := \xi_{K^{(1)}_{n}}\gamman \xi_{K^{(1)}_{n}}$ and $\gamma^{(2)}_{n} := \zeta_{K^{(2)}_{n}}\gamman \zeta_{K^{(2)}_{n}}$ satisfy 

\begin{flalign}[left = \empheqlbrace\,]
   &  \rho_{\gamman} = \rho_{\gamma^{(1)}_{n}}~\text{on } B_{K_{1,n}(y_n)}, \label{eq:dichotomy_begin}
    \quad 
     \rho_{\gamman} = \rho_{\gamma^{(2)}_{n} }~\text{on } B^{c}_{K_{2,n}(y_n)},&  \\
    & \limit{n}{\infty}\tr \gamma^{(1)}_{n} = \alpha,  \\
    & \limit{n}{\infty} \tr \gamma^{(2)}_{n} = \lambda -  \alpha,&  \\
     &\rho_{\gamma^{(1)}_{n}} + \rho_{\gamma^{(2)}_{n}} - \rho_{\gamman} \gegen{n}{\infty} 0 \text{ in } L^p \text{ for all } p \in [1,3),& \label{eq:dichotomy_splitting}\\
     & \norm{\rho_{\gamma_{n}}}_{L^p\left(B_{K^{(2)}_n}(y_n) \setminus \overline{B}_{K^{(1)}_n}(y_n)\right)} \gegen{n}{\infty} 0 \text{ in } L^p \text{ for all } p \in [1,3), & \label{eq:dichotomy_annulus} \\
     &\limit{n}{\infty} \mathrm{dist}\big( \supp\big( \rho_{\gamma^{(1)}_{n}} \big), \supp\big( \rho_{\gamma^{(2)}_{n}} \big)\big) = \infty, & \label{eq:dichotomy_distance} \\
     &\liminfn \tr \big[- \Lap\big( \gamma_n - \gamma^{(1)}_{n}  -\gamma^{(2)}_{n}  \big) \big] \geq 0.  \label{eq:dichotomy_kinetic}
\end{flalign}

In terms of the energy functional this splitting gives 
\begin{align*}
    \E^{X_2}[\gamman]
    &=
    \E^{\infty}[\gamma^{(1)}_{n} ]
    +
    \E^{\infty}[\gamma^{(2)}_{n} ] 
    +
    \intr \rho_{\gamma^{(1)}_{n} } V^{X_2} 
     + 
     \intr \rho_{\gamma^{(2)}_{n} } V^{X_2}  
    + \intr \tilde{\rho}_n V^{X_2}\\ 
    \phantom{= } 
    &+ \tr \big[- \Lap\big( \gamma_n - \gamma^{(1)}_{n}  -\gamma^{(2)}_{n}  \big) \big]\\
   \phantom{=}&+ D[\rho_{\gamma^{(1)}_{n} },\rho_{\gamma^{(2)}_{n} }] + D[\tilde{\rho}_n, \rho_{\gamma^{(1)}_{n} } + \rho_{\gamma^{(2)}_{n} }] + J[\tilde{\rho}_n] \\
   \phantom{=}&+ \intr e_{xc}\big(\rhon \big) - e_{xc}\big(\rho_{\gamma^{(1)}_{n} } \big) - e_{xc}\big(\rho_{\gamma^{(2)}_{n} }\big) \di x,
\end{align*}{}
where we have denoted $\tilde{\rho}_n = \rhon - \rho_{\gamma^{(1)}_{n} } -\rho_{\gamma^{(2)}_{n} }$.
Since \eqref{eq:dichotomy_splitting} we know $\tilde{\rho}_n$ converges to zero in $L^p(\R^3)$ for all $p \in [1,3)$, so we obtain
\begin{align*}
    \intr \tilde{\rho}_n V^{X_2}
    +
     D[\tilde{\rho}_n,\rho_{\gamma^{(1)}_{n} } + \rho_{\gamma^{(2)}_{n} }] 
     + J[\tilde{\rho}_n]
     \gegen{n}{\infty} 0.
\end{align*}{}
Indeed, the first term can be handled by again splitting up the Coulomb potential and the second and third term are dealt with using Hardy-Littlewood-Sobolev
\[
\big|D[f,g] \big| = \bigg| \intr \intr \frac{f(x) g(y)}{|x-y|} \di x \di y \bigg| \leq C \norm{f}_{L^{\nicefrac{6}{5}}}  \norm{g}_{L^{\nicefrac{6}{5}}} .
\]
Furthermore for the Coulomb-interaction between $\rho_{\gamma^{(1)}_{n} }$ and $\rho_{\gamma^{(2)}_{n} }$ we have
\begin{align*}
    D[\rho_{\gamma^{(1)}_{n} }, \rho_{\gamma^{(2)}_{n} }]
    \leq 
   \mathrm{dist}\big( \supp\big( \rho_{\gamma^{(1)}_{n}}  , \rho_{\gamma^{(2)}_{n}} \big) \big)^{-1} \norm{\rho_{\gamma^{(1)}_{n} }}_{L^1} \norm{\rho_{\gamma^{(2)}_{n} }}_{L^1} 
    \gegen{n}{\infty} 0,
\end{align*}{}
where we used \eqref{eq:dichotomy_distance}.
Also the difference in the exchange terms vanishes

\begin{align*}
    \bigg|\intr e_{xc}\big(\rhon \big) &- e_{xc}\big(\rho_{\gamma^{(1)}_{n} } \big) - e_{xc}\big(\rho_{\gamma^{(2)}_{n} } \big) \bigg| \\
    &\leq 
    \int\limits_{B_{K_{2,n}}(y_n) \backslash \overline{B}_{K_{1,n}}(y_n)} \big|e_{xc}\big(\rhon \big)\big| + \big|e_{xc}\big(\rho^{H}_{\gamma^{(1)}_{n}} \big)\big| + \big|e_{xc}\big(\rho^{H}_{\gamma^{(2)}_{n}} \big) \big| \\
    &\leq 
    3C \bigg(\norm{\rhon}^{p_{-}}_{L^{p_{-}}\big( B_{K_{2,n}}(y_n) \backslash \overline{B}_{K_{1,n}}(y_n)\big) }
    +\norm{\rhon}^{p_{+}}_{L^{p_{+}}\big( B_{K_{2,n}}(y_n) \backslash \overline{B}_{K_{1,n}}(y_n)\big) }\bigg) \gegen{n}{\infty} 0,
\end{align*}{}
where the exponents $p_\pm$ are given by $p_{\pm} = 1 + \beta_\pm$.
Indeed here we employed Assumption\ref{assumption1} part 3 on $e_{xc}'(\rho)$ together with part 2 $e_{xc}(0)=0$ in order to get 
\[
\big| e_{xc}(\rho)  \big| \leq C \big( \rho^{p_+} + \rho^{p_-}\big)
\] and equation \eqref{eq:dichotomy_annulus}, i.e. that the density vanishes on the annulus $B_{K_{2,n}}(y_n) \backslash \overline{B}_{K_{1,n}}(y_n)$.

Using the \(\liminf \) estimate for the kinetic energy from \eqref{eq:dichotomy_kinetic}, we obtain
\begin{align} \label{energy_cutoff}
      \E^{X_2}[\gamman]
      \geq 
       \E^{ \infty}[\gamma^{(1)}_{n}]
    +
    \E^{ \infty}[\gamma^{(2)}_{n}] 
    +
    \intr \rho_{\gamma^{(1)}_{n}} V^{X_2} 
     + 
     \intr \rho_{\gamma^{(2)}_{n}} V^{X_2}  
          + \mathcal{R}(n)
\end{align}
with a remainder \(\mathcal{R}(n) \gegen{n}{\infty} 0 \).
The last step is to deal with the nuclei part and to go from $V^{X_2}$ to $V^{X}$; here we again have to distinguish three cases.

\underline{Case1:} $\rhogn{1}$ stays close to exactly one nucleus (w.l.o.g. the one at the origin), i.e.
\begin{align*}
    \mathrm{dist} \big( {0}, B_{K^{(1)}_{n}}(y_n)\big) \text{ stays bounded and }
    \mathrm{dist} \big( {R_n}, B_{K^{(1)}{n}}(y_n)\big) \gegen{n}{\infty} \infty.
\end{align*}{}
Note that this necessarily implies that $\mathrm{dist}\big( 0, \supp(\rhogn{2})\big) \to \infty$ due to triangle inequality.
Hence,
\begin{align*}
    \intr \rhogn{1} \frac{1}{|x - R_n|} \di x, ~   
    \intr \rhogn{2} \frac{1}{|x|} \di x
    \gegen{n}{\infty}0 
\end{align*}{}
and thus taking the limit in  \eqref{energy_cutoff} and using the continuity of $\lambda \mapsto I^{X}_\lambda$ gives 
\begin{align*}
    \liminfn 
    \In = 
    \liminfn
    \E^{X_2}[\gamman]
      \geq
      \liminfn 
       \E^{X}[\gamma^{(1)}_n]
    +
    \E^{X}[\gamma^{(2)}_n] 
        \geq 
    I^{X}_{\alpha} 
    + 
     I^{X}_{\lambda - \alpha} 
     \geq \minimum.
\end{align*}{}

\underline{Case 2:}  $\rhogn{1}$ does not stay close to any of the two nuclei, i.e.
\begin{align*}
    \mathrm{dist} \big(\{0, R_n \}, \supp\big( \rhogn{1} \big) \big) \gegen{n}{\infty} \infty.
\end{align*}{}
Then by using again our upper bound \eqref{eq:upper_bound} from Subsection \ref{sec:upper_bound} equation \eqref{energy_cutoff} becomes
\begin{align*}
    \minimum
    \geq 
    \liminfn 
    \E^{X_2}[\gamman]
    &\geq 
    I^{\infty}_{\alpha}
    +
    \liminfn
    \bigg(\E^{\infty}[\gamma^{(2)}_{n}] 
    +
    \intr \rhogn{2} V^{X_2}  \bigg)\\
    &
    = 
    I^{\infty}_{\alpha}
    +
    \liminfn \E^{X_2}[\gamma^{(2)}_{n}]\\
    &
    \geq 
    I^{ \infty}_{\alpha}
    +
    \underbrace{\liminfn \E^{X_2}[ \tilde{\gamma}_{n}]}_{=:J_{\lambda - \alpha}},
\end{align*}{}
where $\tilde{\gamma}_n$ is a minimizer of the problem $I^{X_2}_{\lambda - \alpha, R_n}$ for each $n$.

\underline{Case 3:}
Here $\rhogn{1}$ stays close to both of the nuclei and hence $\rhogn{2}$ does not stay close to any of the two. Thus we get the same result as in case 2, but with $\alpha$ and $\lambda - \alpha$ exchanged.

So we obtain 
\begin{align*}
    J_{\lambda} \geq I^{\infty}_{\alpha} + J_{\lambda - \alpha} \text{ (case 2)}
    ~~\text{ or }~~
    J_{\lambda} \geq I^{\infty}_{\lambda - \alpha} + J_{\alpha} \text{ (case 3)}.
\end{align*}
Note furthermore that  the opposite inequality always holds. As in the proof of the upper bound (section \ref{sec:upper_bound}) we can for any $\epsilon >0$ take finite rank approximations with range in $C^\infty_c(\R^3)$ of the minimizers $\gamma^\infty_\alpha$ to $I^\infty_\alpha $ and $\gamma_{\lambda - \alpha, n}$ to $J_{\lambda - \alpha}$, respectively, to obtain
\begin{align*}
    J_{\lambda} \leq \E^{X_2} [ \gamma_{\lambda - \alpha, n} + \tau_R \gamma^\infty_\alpha \tau_{-R} ] 
    \leq J_{\lambda - \alpha} + I^\infty_\alpha + \epsilon + O(\tfrac{1}{R}) \gegen{\epsilon, \tfrac{1}{R}}{0}  J_{\lambda - \alpha} + I^\infty_\alpha 
    ,
\end{align*}
so we arrive at 
\begin{align} \label{eq:splittoff_minimizer}
    J_{\lambda} = I^{\infty}_{\alpha} + J_{\lambda - \alpha} \text{ (case 2)}
    ~~\text{ or }~~
    J_{\lambda} = I^{\infty}_{\lambda - \alpha} + J_{\alpha} \text{ (case 3)}.
\end{align}
This furthermore implies that the part splitting off to infinity (i.e.~$\gamma^{(1)}_n$ in case 2 and $\gamma^{(2)}_n$ in case 3) is almost a minimizing sequence for the problem at infinity ($I^\infty_\alpha$ in case 2 and $I^\infty_{\lambda - \alpha}$ in case 3) in the sense that
\begin{align} \label{eq:almost_minimizing}
\limit{n}{\infty} \E^\infty[\gamma^{(1)}_n] = I^\infty_\alpha 
\quad 
\text{and}
\quad
\limit{n}{\infty} \tr[\gamma^{(1)}_n] = \alpha.
\end{align}

Now we are in the same position as in the beginning of the proof: We have a sequence $\tilde{\gamma}_n$ of minimizers to the functional $\E^{X_2}_{R_n}$ but now with the mass constraint 
\begin{equation} \label{eq:finally_positiv}
    \tr [\tilde{\gamma}_n] = \lambda -\alpha < \lambda \leq N.
\end{equation}{}
Going through the entire procedure of the proof again, it either ends after a finite amount of steps $(i)$ or we always end up into the dichotomy case and there case 2 or 3 $(ii)$ compare Figure \ref{fig:structure_proof}.

\underline{Case $(i)$:}
After a finite amount of steps we get
\begin{align} \label{eq:finite_splitting}
    J_\lambda
    \geq \sum_{l=1}^{T} I^{\infty}_{\alpha_{l}} + J_{\lambda - \sum_{l=1}^{T} \alpha_l}
    \geq 
    \sum_{l=1}^{T} I^{\infty}_{\alpha_{l}} + \min\limits_{\beta } \big(  I^{X}_{\beta} + I^{X}_{\lambda - \beta - \sum_{l=1}^{T} \alpha_l} \big). 
\end{align}
Let the minimum at the right hand side be attained at $\tilde{\beta}$, then we can just apply the weak subadditivity inequality from Lemma \ref{lem:subadditivity} to obtain the desired assertion
\begin{align*}
    \eqref{eq:finite_splitting}
     \geq \sum_{l=1}^{T} I^{\infty}_{\alpha_{l}} +  I^{X}_{\tilde{\beta}} + I^{X}_{\lambda - \tilde{\beta} - \sum_{l=1}^{T} \alpha_l}
    \geq I^{X}_{\tilde{\beta}} + I^{X}_{\lambda - \tilde{\beta}}
    \geq \minimum.
\end{align*}

\underline{Case $(ii)$:}
This case is the more intricate one. 


After the first splitting we have for the sequence $\tilde{\gamma}_n$ that $\norm{\rho_{ \tilde{\gamma}_n}}_{L^1} = \lambda - \alpha < \lambda$. 
In order to show that such a splitting can not occur infinitely many times we start with considering the Euler-Lagrange equations of the system.

\begin{lemma}[Euler-Lagrange equations]
Let $\gamma_R$ be a minimizer to the energy functional $\E^{X^2}_{R}$ to the mass constraint $\tr[\gamma_R] = \tilde{\lambda} < \lambda$ then it satisfies the Euler-Lagrange equations
\begin{equation}\label{eq:euler_lagrange}
\gamma_R = \mathds{1}_{(-\infty, \epsilon_F)} \big( h_{\gamma_R}\big) + \delta, \quad
\text{with } 0 \leq \delta \subset \kernel \big( h_{\gamma_R} - \epsilon_F \big) 
\end{equation}{}
for some $\epsilon_F <0$ called the Fermi energy, and with the Hamiltonian 
\[
h_{\gamma_R} 
= \big( - \tfrac{1}{2} \Lap + \rho_{\gamma_R} \ast \tfrac{1}{|x|} + V^{X_2}_{R} + e_{xc}'(\rho_{\gamma_R}) \big).
\]
Furthermore, we have
\begin{equation} \label{eq:euler_lagrange_occupation}
    \gamma_R \in \argmin\{ \tr[h_{\gamma_R} \gamma] : ~ \gamma \in K_\lambda \}.
\end{equation}{}
\end{lemma}{}
\begin{proof}
This is a standard result, but let us shortly proof it (for a more detailed version see \cite{Gontier_2014}). If $\gamma_R$ is a minimizer for $\E^{X_2}_R$ with $\tr[\gamma_R]=\lambda$ we have for any $\gamma \in K_\lambda$ the inequality $\E^{X_2}_R[t\gamma + (1-t)\gamma_R]\geq \E^{X_2}_R[\gamma_R]$. In particular
\begin{align}\label{eq:euler-lagrange-gontier}
\diffp{}{t}\E^{X_2}_R[t\gamma + (1-t)\gamma_R] \bigg\lvert_{t=0} \geq 0.
\end{align}
A direct calculation leads to 
\[
\diffp{}{t}\E^{X_2}_R[t\gamma + (1-t)\gamma_R] \bigg\lvert_{t=0}
=
\tr\big[ h_{\gamma_R} (\gamma - \gamma_R) \big]
\]
with $h_{\gamma_R}$ as above.
Due to \eqref{eq:euler-lagrange-gontier} we must have $ \gamma_R \in \argmin\{ \tr[h_{\gamma_R} \gamma] : ~ \gamma \in K_\lambda \}$.

The representation of $\gamma_R$ then follows if we can show $\epsilon_F <0$. 
But as $ \rho_{\gamma_R} \ast \tfrac{1}{|x|} + V^{X_2}_{R} + e_{xc}'(\rho_{\gamma_R}) $ is $\Delta$-compact, since it is in $L^\infty_\epsilon + L^2$, by Weyl theorem \cite{reed4} the essential spectrum is that of the Laplacian $\sigma_{ess} (h_{\gamma_R}) = [0,\infty)$. Furthermore $h_{\gamma_R}$ is bounded from below and due to $e'_{xc}(x) \leq 0$ we have the bound
\begin{align}
    h_{\gamma_R} \leq \big( - \tfrac{1}{2} \Lap + \rho_{\gamma_R} \ast \tfrac{1}{|x|} + V^{X_2}_{R} \big).
\end{align}
For the operator on the right hand side we know by Lemma 19 from \cite{lions_hf} that as long as the nuclear charge $2Z$ is larger than $\tilde{\lambda}$, which is satisfied, since $2z \geq \lambda > \tilde{\lambda}$, it has infinitely many negative eigenvalues of finite multiplicity.
Therefore the same holds true for $h_{\gamma_R}$, which gives us $\epsilon_F <0$.
\end{proof}{}

Note that \eqref{eq:euler_lagrange_occupation} implies that in fact only finitely many orbitals are occupied, i.e.
\[
\tilde{\gamma}_n = \sum_{l=1}^k | \phi^l_n \rangle\langle \phi^l_n | + \sum_{l=k}^m \lambda_l | \phi^l_n \rangle\langle \phi^l_n |,
\]
with $\lambda_l \in [0,1]$. Here the first $n$ orbitals are fully occupied, while the rest might be fractionally occupied.

Furthermore every occupied orbital $\phi^l_n$ is an eigenstate of the corresponding hamiltonian $h_{\tilde{\gamma}_n}$, i.e.~satisfies
\begin{align} \label{eq:el_first_splitting}
    \big( - \tfrac{1}{2} \Lap + \rho_{\tilde{\gamma}_n} \ast \tfrac{1}{|x|} + V^{X_2}_{R_n} + e_{xc}'(\rho_{\tilde{\gamma}_n}) \big) \phi^l_{n}  + \theta^l_n \phi^l_{n} = 0,
\end{align}{}
where $-\theta^1_n < -\theta^2_n \leq \ldots$ denotes the ordered eigenvalues.
Our first step consist in proving that for fixed $l$ the sequence $\big(\theta^l_n\big)_n$ stay bounded away from 0.
\begin{lemma}\label{lem:theta_negativ}
Denote by $\big(\theta^l_n\big)_n$ the sequence of smallest eigenvalues in \eqref{eq:el_first_splitting}, then we have
\begin{equation}
    \liminfn \theta^l_n   >0.
\end{equation}{}
\end{lemma}
\begin{proof}
To see this note 
\begin{align*}
    h_{\rho_{\tilde{\gamma}_n}} \leq - \tfrac{1}{2} \Lap + \rho_{\tilde{\gamma}_n} \ast \tfrac{1}{|x|} + V^{X_2}_{R_n} = \tilde{h}_n,
\end{align*}{}
so it is enough to consider the latter operator $ \tilde{h}_n$.
As in \cite{lions_hf} consider a radially symmetric function $\psi \in C^{\infty}_{c}$ with $\norm{\psi}_{L^2} =1$ and
set $\psi_\sigma = \sigma^{\nicefrac{3}{2}} \psi(\sigma \cdot)$.
Then we get 
\begin{align*}
    \langle \psi_\sigma,  \tilde{h}_n \psi_\sigma \rangle 
    =
     \sigma^2 \frac{1}{2} \intr |\nabla \psi|^2 \di x
    +
    \sigma  \intr V_\sigma(x) |\psi|^2 \di x
    + 
   \sigma  \intr \bigg( \rho_{\sigma, \tilde{\gamma}_n} \ast \tfrac{1}{|x|}\bigg) |\psi|^2 \di x,
\end{align*}{}
where $V_\sigma (x) = -\tfrac{Z}{|x|} - \tfrac{Z}{\big|x-{R_n}{\sigma}\big|}$ and $\rho_{\sigma, \tilde{\gamma}_n} = \sigma^{-3} \rho_{\tilde{\gamma}_n}(\frac{1}{\sigma}\cdot)$.
Note that the $\sigma^2$ in front of the kinetic energy comes from the chain rule while the prefactors of the two remaining terms come from substitution.
Due to radial symmetry of $\psi$ we have
\begin{align*}
     \intr \bigg( \rho_{\sigma, \tilde{\gamma}_n} \ast \tfrac{1}{|x|}\bigg) |\psi|^2 \di x
     &=
     \intr \bigg(|\psi|^2  \ast \tfrac{1}{|x|}\bigg)\rho_{\sigma, \tilde{\gamma}_n}(x)  \di x
     \\
     &=
     \intr \intr \frac{|\psi|^2(y)}{\max\{|x|,|y|\}} \di y \, \rho_{\sigma, \tilde{\gamma}_n}(x) \di x\\
     &\leq 
    \underbrace{ \norm{\rho_{\sigma, \tilde{\gamma}_n}}_{L^1}}_{\lambda - \alpha }  \intr \frac{|\psi|^2(y)}{|y|} \di y.
\end{align*}{}
By Rayleigh-Ritz we thus have for every fixed $n$ 
\begin{align*}
    - \theta^1_n = \inf \langle \psi_\sigma , h_{\rho_{\tilde{\gamma}_n}}\psi_\sigma \rangle
    \leq
    \inf \langle \psi_\sigma , \tilde{h_{n}}\psi_\sigma  \rangle    
    \leq \sigma \underbrace{\big( \lambda - \alpha  -2Z \big)}_{<0} \intr \frac{|\psi|^2(y)}{|y|} \di y + \sigma^2  \frac{1}{2} \intr |\nabla \psi|^2 \di x.
\end{align*}{}
Taking  now $\sigma \to 0$ the linear term will eventually dominate and since the right hand side is independent of $n$, we obtain 
\[
\theta_n^1 > 0  \quad \forall ~n\in \N.
\]
For $l>1$ simply take a family of orthogonal functions $\big( \psi_j \big)_{j=1}^k$ with the same properties as $\psi$ above, the min-max principle \cite{reed4} then gives the result.
\end{proof}{}

Since also for $\tilde{\gamma}_n$ the dichotomy case occurs we get $\tilde{\gamma}^{(1)}_n$ and $\tilde{\gamma}^{(2)}_n$ with the same properties as listed in \eqref{eq:dichotomy_begin} - \eqref{eq:dichotomy_kinetic}.  
Define now 
\[\omega^l_n : = (1- \xi_{K^{(1)}_n} - \zeta_{K^{(2)}_n}) \phi^l_n = \epsilon_n \phi^l_n 
\quad 
\text{and} \quad \phi^l_{1,n} = \xi_{K^{(1)}_n} \phi^l_n, ~~\phi^l_{2,n} =  \zeta_{K^{(2)}_n} \phi^l_n,\]
where $\phi_n^l$ are the orbitals corresponding to $\tilde{\gamma}_n$ and $\xi, \zeta$ are the smooth partition of unity given by the
dichotomy case with $K^{(j)}_n$ as in \eqref{eq:dichotomy_increasing}.

Note that $0 \leq \epsilon_n \leq 1$ and $\norm{\nabla \epsilon_n}_\infty \to 0$.
Furthermore we have 
\[
\rho_{\tilde{\gamma}_{n}^{(i)}} =  \sum_{l} \lambda_l^{(n)} |\phi^l_{i,n}|^2,
\]
where $0 \leq \lambda^{(n)}_l \leq 1$ is the occupation number of the $l^{th}$ orbital.
By multiplying \eqref{eq:el_first_splitting} with $\omega^l_n$, we obtain 
\begin{align*}
    \intr \nabla \omega^l_n \cdot \nabla \phi^l_n \di x \gegen{n}{\infty} 0.
\end{align*}{}
Since $\nabla \omega^l_n = \epsilon_n \nabla \phi^l_n + \phi^l_n \nabla \epsilon_n$ and $\epsilon_n^2 \leq \epsilon_n$ we also get
\begin{align*}
    \intr \epsilon_n^2|\nabla \phi^l_n|^2 \di x \gegen{n}{\infty} 0,
\end{align*}{}
which finally implies $\nabla \omega^l_n \to 0$ in $L^2(\R^3).$
Combining this with the fact that the supports of $\phi^l_{1,n}$ and $\phi^l_{2,n}$ go infinitely far apart for $n \to \infty $ \eqref{eq:el_first_splitting} becomes
\begin{align} \label{eq:splitting1}
    \big( - \tfrac{1}{2} \Lap + \rho_{\tilde{\gamma}^{(1)}_n } \ast \tfrac{1}{|x|} + V^{X_2}_{R_n} + e_{xc}'( \rho_{\tilde{\gamma}^{(1)}_n }) \big) \phi^l_{1,n}+ \theta^l_n  \phi^l_{1,n} \xrightarrow[H^{-1}]{n \to \infty}0\\
    \label{eq:splitting2}
    \big( - \tfrac{1}{2} \Lap + \rho_{\tilde{\gamma}^{(2)}_n} \ast \tfrac{1}{|x|} + V^{X_2}_{R_n} + e_{xc}'( \rho_{\tilde{\gamma}^{(2)}_n}) \big) \phi^l_{2,n} + \theta^l_n  \phi^l_{2,n} \xrightarrow[H^{-1}]{n \to \infty}0
\end{align}{}
Note here that the eigenvalues $\theta^l_n$ are the ones from $h_{{\tilde{\gamma}_n}}$ and that the support of one of the two sequences drifts infinitely far way of both nuclei.
W.l.o.g.~let it be $\tilde{\gamma}_{n}^{(1)}$, then 
\[
\dist\big(\{0,R_n\}, \rho_{\tilde{\gamma}_{n}^{(1)}} \big) \gegen{n}{\infty} \infty
\]
and since  $\tilde{\gamma}_{n}^{(1)}$ is almost a minimizing sequence to $I^{\infty}_\alpha$ in the sense of \eqref{eq:almost_minimizing}, it cannot vanish.
Therefore there exists $\kappa, M >0$ and a sequence $(y_n)_n$ of points in $\R^3$ such that 
\begin{align} \label{eq:nonvanishing}
\int_{B_M(y_n)} \rho_{\tilde{\gamma}_{n}^{(1)}}(x) \di x \geq \kappa >0 .
\end{align}{}
Furthermore we necessarily have 
\[
\dist\big(\{0,R_n\}, (y_n)_n\big) \gegen{n}{\infty} \infty
\]
and thus \eqref{eq:splitting1} becomes for the translated density matrix $\Bar{\gamma}_{n}^{(1)} : = \tau_{y_n} \tilde{\gamma}_{n}^{(1)} \tau_{-y_n}$ with orbitals $\Bar{\phi}^l_{1,n}$
\begin{align*}
    \big( - \tfrac{1}{2} \Lap + \rho_{\Bar{\gamma}_{n}^{(1)} } \ast \tfrac{1}{|x|}+ e_{xc}'( \rho_{\Bar{\gamma_{n}}^{(1)}}) \big) \Bar{\phi}^l_{1,n}+ \theta^l_n  \Bar{\phi}^l_{1,n} \xrightarrow[H^{-1}]{n \to \infty}0.
\end{align*}{}
Finally note that $\sqrt{\rho_{\Bar{\gamma}_{n}^{(1)}}} \rightharpoonup \sqrt{\rho} \neq 0$ in $H^1(\R^3)$ due to \eqref{eq:nonvanishing}.

Now if our procedure never stops we can as in \cite{cances} construct an infinity of sequences of orbitals $\big( \phi^l_{k,n}\big)_{k,n \in \N}$ with $\norm{ \phi^l_{k,n}}_{L^2}=1$ such that for every $k,n \in \N$

\begin{flalign}[left = \empheqlbrace\,] \label{eq:ltwonorminidep}
  &  \psi_{l,k,n} := \big(\sqrt{\lambda_l^{(n)}} \phi^l_{k,n} \big), \sqrt{\rho_{\gamma^{(k)}_n}} \text{ bdd. in } H^{1}(\R^3),  \intr \rho_{\gamma^{(k)}_n} = \alpha_k, ~ \rho_{\gamma^{(k)}_n} = \sum_{l} |\psi_{l,k,n}|^2&  \\
  & \big( - \tfrac{1}{2} \Lap +  \rho_{\gamma^{(k)}_n} \ast \tfrac{1}{|x|}  + e_{xc}'( \rho_{\gamma^{(k)}_n}) \big) \phi^l_{k,n}  + \theta^l_n \phi^l_{k,n} = \eta_{n} \xrightarrow[n \to \infty]{H^{-1}} 0 &  \label{eq:ekeland1}\\
  & \psi_{l,k,n}  \text{ converges to } \psi_{l,k} \text{ weakly in } H^1, \text{ strongly in } L^{p}_{loc} \text{ for } 2 \leq p <6 \text{ and a.e.~on } \R^3,  \\
  &\sqrt{\rho_{\gamma^{(k)}_n}}  \text{ converges to } \sqrt{\rho_{k}} \neq 0 \text{ weakly in } H^1, \text{ str. in } L^{p}_{loc} \text{ for } 2 \leq p <6 \text{ and a.e.~on } \R^3,  &
\end{flalign}{}
where 
\begin{align} \label{eq:limit_weak_mass}
    \sum_{k \in \N} \alpha_k \leq \lambda- \alpha.
\end{align}{}
Note furthermore that
\begin{align*}
    \sum_{l} \norm{ \psi_{l,k,n}}_{H^1}^2
\end{align*}{}
stays bounded independent of $k$ or $n$. For the $L^2$-norm this is clear from \eqref{eq:ltwonorminidep}, since $\alpha_k$ stays bounded \eqref{eq:limit_weak_mass}.
For the $H^1$-norm, this can be seen by applying \eqref{eq:ekeland1} to $\lambda^{(n)}_l \phi^l_{k,n}$ and summing over $l$.
Then the term
\begin{align*}
   \bigg ( \frac{1}{2}\sum_l   \norm{\nabla \psi_{l,k,n}}^2_{L^2}  + J[\rho_{\gamma_n^{(k)}} ]  + \intr e_{xc}'(\rho_{\gamma_n^{(k)}})\rho_{\gamma_n^{(k)}} \di x + \sum_{l} \theta^l_n \norm{\psi_{l,k,n}}_{L^2}^2  \bigg).
\end{align*}
stays bounded independently of $k$ and $n$. But by \eqref{eq:ltwonorminidep} $\sqrt{\rho_{\gamma_n^{(k)}}}$ is bounded in $H^1$, thus in $L^p$ for $2\leq p \leq 6$. So each of the last three terms stays bounded, because $J$ is bounded by Hardy-Littlewood-Sobolev
\[J[\rho] \leq C \norm{\rho}_{L^{\frac{6}{5}}}^2 < \infty, \]
the exchange-correlation term by 
\[
\intr e_{xc}'(\rho)\rho \di x \leq 
C \intr \sqrt{\rho}^{2+2\beta_-} + \sqrt{\rho}^{2+2\beta_+}
\]
with $2+2\beta_{\pm} \in [2, 6]$ and the last term by the boundedness of $\theta^l_n$ and the already discussed $L^2$-norm bound.
Thus taking the limit $n \to \infty$ we get  
\begin{align}\label{eq:final_pde}
     \big( - \tfrac{1}{2} \Lap + \rho_{k} \ast \tfrac{1}{|x|}  + e_{xc}'(\rho_{k}) \big) \psi_{l,k}  + \theta^l \psi_{l,k} = 0,
\end{align}{}
where $\theta^l = \liminf \limits_{n \to \infty} \theta^l_n >0$.
Furthermore we have
\begin{equation} \label{eq:finite_sum_orbitals}
\rho_{k} = \sum_{l} |\psi_{l,k}|^2  . 
\end{equation}
Since the mass of the $\rho_{\gamma^{(k)}_n}$ does not depend on $n$ we obtain from \eqref{eq:limit_weak_mass}
\begin{align}\label{eq:rhok_to_zero}
    \limit{k}{\infty} \norm{\rho_{k}}_{L^1} = 0.
\end{align}{}
By multiplying \eqref{eq:final_pde} with $\psi_{l,k}$, integrating and summing over $l$ we obtain
\begin{align*}
    0 &\geq 
    - \sum_l \theta^l \norm{\psi_{l,k}}^2_{L^2} -  \intr \bigg(\rho_k \ast \frac{1}{|x|}) \bigg) |\psi_{l,k}|^2 \di x  \\
    &=
    \frac{1}{2} \sum_{l} \norm{\nabla \psi_{l,k}}^2_{L^2} 
    + \sum_{l} \intr |\psi_{l,k} |^2 e_{xc}'(\rho_k) \di x.
\end{align*}
Now we can again use Assumption\ref{assumption1} stating $e_{xc}'(\rho) \leq C (\rho^{\beta_-} + \rho^{\beta_+})$ on the second term and then apply Hölders inequality with $\beta_+$ and $\beta_-$, respectively, to obtain the following bound 
\begin{align*}
0
     &\geq 
      \frac{1}{2} \sum_{l} \norm{\nabla \psi_{l,k}}^2_{L^2} 
      - C \sum_{l} \big( \norm{\psi_{l,k}^2}_{L^\frac{1}{1- \beta_{-}}} + \norm{\psi_{l,k}^2}_{L^\frac{1}{1-\beta_{+}}} \big) \norm{\rho_k}_{1}\\
      &\geq 
      \frac{1}{2} \sum_{l} \norm{\nabla \psi_{l,k}}_{L^2}^2 
      - C  \norm{\rho_k}_{L^1} \underbrace{ \sum_{l} \norm{ \psi_{l,k,n}}_{H^1}^2}_{<\infty},
\end{align*}{}
where we applied the Sobolev embedding in the last inequality on the two terms $\norm{\psi^2}_{L^\frac{1}{1-\beta}} = \norm{\psi}_{L^\frac{2}{1-\beta}}^2$. Indeed, by Assumption\ref{assumption1} we have $0 <\beta< \frac{2}{3}$, implying $\frac{2}{1-\beta} \in (2,6)$.
Thus by taking the limit $ k \to \infty$ and using \eqref{eq:rhok_to_zero} we obtain
\begin{align*}
    \sum_{l} \norm{\nabla \psi_{l,k}}_{L^2}^2 \gegen{k}{\infty} 0.
\end{align*}{}
Applying standard elliptic regularity results (see e.g.~\cite{gilbarg_trudinger}) to \eqref{eq:final_pde} now give us the inequality
\begin{align*}
    \norm{\psi_{l,k}}_{L^\infty}
    \leq 
    C
    \norm{\psi_{l,k}}_{H^1},
\end{align*}{}
where the constant $C>0$ does not depend on $k$ and thus
\begin{align*}
    \limit{k}{\infty}\sum_{l} \norm{ \psi_{l,k}}^2_{L^\infty} = 0.
\end{align*}{}
Thus by \eqref{eq:finite_sum_orbitals} we also obtain
\begin{align*}
    \limit{k}{\infty}\norm{ \rho_{k}}_{L^\infty} = 0.
\end{align*}{}
Again from \eqref{eq:final_pde} and from Assumptions \ref{assumption1} we deduce
\begin{align}\label{eq:final_eq_zero}
    \theta^l \norm{\psi_{l,k}}_{L^2}^{2} \leq C \big( \norm{\rho_k}_{L^\infty}^{2 \beta_{-}} + \norm{\rho_k}_{L^\infty}^{2 \beta_{+}} \big) \norm{\psi_{l,k}}_{L^2}^2.
\end{align}{}
Now note that due to \eqref{eq:euler_lagrange_occupation} at most $N$ different energy levels are occupied. Thus
\begin{align*}
    \norm{\psi_{l,k,n}}_{L^2}^2 = \lambda_l^n = 0,
\end{align*}{}
for all $l$ corresponding to the $(N+1)^{\text{th}}$ or higher eigenvalues without counting multiplicity. Note that due to degeneracies this might not be the same as $l>N$.
Therefore we directly get for those $l$
\[
\norm{\psi_{l,k}}_{L^2} = 0.
\]
Thus the mass of $\rho_k$ is distributed among only finitely many energy levels $l$.
Therefore for at least one fixed level $l$ we can find up to a subsequence in $k$ $\psi_{l,k}$ such that 
\[
\norm{\psi_{l,k}}_{L^2} \neq 0, ~~ \forall k,
\]
because otherwise we would have $\norm{\rho_k}_{L^1} =0$.
Hence \eqref{eq:final_eq_zero} becomes
\begin{align*}
    \theta^l \leq  C \big( \norm{\rho_k}_{L^\infty}^{2 \beta_{-}} + \norm{\rho_k}_{L^\infty}^{2 \beta_{+}} \big) 
    \gegen{k}{\infty}0,
\end{align*}{}
which is a contradiction to Lemma \ref{lem:theta_negativ}.

Thus case (ii) can not happen and the proof is hence complete.

\subsection{Proof of Proposition \ref{prop:linear_case} \label{sec:proof_linear_case}}
First we will show that we can get ride of the antisymmetry condition, i.e.
\begin{align*}
    I_R^{H_2} =  \inf \limits_{\substack{\psi \in H^1(\R_\Sigma^2),\\ \norm{\psi}_{L^2}=1,\\ \psi \text{ antisymm.}}} \langle \psi , H(x,y) \psi \rangle
     =
         \inf \limits_{\substack{\psi \in H^1(\R^2),\\ \norm{\psi}_{L^2}=1}} \langle \psi , H(x,y) \psi \rangle = : \tilde{I}_R^{H_2}
\end{align*}

Take any $\psi \in H^1(\R_\Sigma^2), \norm{\psi}_{L^2}=1$, then
\begin{align*}
    \langle \psi, H(x,y) \psi \rangle 
    &=
    \sum_{s,t\in \Sigma} \langle \psi(\cdot,s,\cdot,t), H(x,y) \psi(\cdot,s,\cdot,t)\rangle \\
    & \geq 
    \sum_{s,t\in \Sigma}  \tilde{I}_R^{H_2}  \norm{\psi(\cdot,s,\cdot,t)}_{L^2(\R^2)}^2 \\
    &=  \tilde{I}_R^{H_2}  \norm{\psi}_{L^2((\R_\Sigma )^2)}^2 =  \tilde{I}_R^{H_2}.
\end{align*}
So we have $ I_R^{H_2} \geq \tilde{I}_R^{H_2}$. For the other direction
define for any given $\psi \in H^1(\R^2)$ with $\norm{\psi}_{L^2}^2$
\begin{align*}
    \tilde{\psi}(x,s,y,t) := \frac{1}{\sqrt{2}}
    \bigg( \psi(x,y) \delta_{|\uparrow\rangle}(s) \delta_{|\downarrow\rangle}(t) - \psi(y,x) \delta_{|\uparrow\rangle}(t) \delta_{|\downarrow\rangle}(s)\bigg),
\end{align*}
where $\delta_{|\uparrow\rangle}(s) $ denotes the Kronecker delta for the spin component.
By definition this $\tilde{\psi}$ satisfies the antisymmetry condition and is normalized.
Furthermore we have
\[
\langle \tilde{\psi}, H(x,y) \tilde{\psi } \rangle_{L^2(\R^2_\Sigma)}
=
\langle \psi, H(x,y) \psi \rangle_{L^2(\R^2)},
\]
taking the infimum over normalized $\psi$ gives 
\[
I_R^{H_2} \leq  \inf_{\tilde{\psi}} \langle \tilde{\psi}, H(x,y) \tilde{\psi } \rangle 
=
\inf_\psi \langle \psi, H(x,y) \psi \rangle = \tilde{I}_{R}^{H_2}.
\]
So from here on we will consider the system without the antisymmetry condition.

\subsubsection{Upper bound}
For the upper bound we can take a $\phi \in C^{\infty}_{c}(\R)$ with $\norm{\phi}_{L^2} =1 $ and consider as a testfunction for the $H_2$ Hamiltonian just the tensor product $\psi = \phi \otimes \phi(\cdot - R)$, i.e.~$\psi(x,y) = \phi(x)\phi(y-R)$.
Then we directly get
\begin{align*}
\limit{R}{\infty}
E^{H^2}_R 
&\leq 
\limit{R}{\infty}
\langle \psi, H_R(x,y) \psi \rangle \\
&= 
\limit{R}{\infty}
2 \langle \phi, h(x) \phi \rangle 
+ |\phi|^2(R) + |\phi|^2(-R) + 2\int  \phi(\pm y) \phi(y-R) \di y \\
&=
2 \langle \phi, h(x) \phi \rangle,
\end{align*}{}
where we used that the last three terms vanish as soon as $R>\mathrm{diam}(\supp\phi)$.
Taking now the infimum w.r.t.~$\phi$ and noting that the Hamiltonian $h$ is continuous on $H^1(\R)$ we get the result.

\subsubsection{Lower bound}
Since the electron-electron interaction is positive we directly get
\begin{align} \label{eq:tensor_structure}
    H(x,y) \geq \tilde{h}(x) + \tilde{h}(y),
\end{align}
where $\tilde{h}(x) = -\frac{1}{2} \diff[2]{}{x} - \delta_0(x) - \delta_R(x)$. 
To determine the infimum over the right hand side, we can just consider tensor-products of functions due to the additive structure.
Hence we only need to consider 
\begin{align*}
    \langle \phi, \tilde{h}(x) \phi \rangle, \quad \phi \in L^2(\R).
\end{align*}
Therefore consider any arbitrary $\phi \in H^1(\R)$, and two cut-off functions $\xi_1$ and $\xi_2$ with 
\begin{align*}
    \xi_1^2 + \xi_2^2 =1, \quad \xi_1(x) = 1 \text{ for } x \leq \tfrac{1}{3}, \quad  \xi_1(x) = 0  \text{ for } x \geq \tfrac{2}{3}.
\end{align*}{}
Defining then 
\begin{align*}
    \phi_i = \xi_i\big(\tfrac{\cdot}{R} \big) \phi,
\end{align*}{}
gives with a straightforward calculation
\begin{align}
    \langle \phi, \tilde{h}(x) \phi \rangle
    &\geq 
    \langle \phi, h_0(x) \phi \rangle
    +\langle \phi, h_R(x) \phi \rangle
    + o(1). \label{eq:proof_propone} 
\end{align}
Here $h_0(x) , h_R(x)$  denote the hamiltonian $h(x)$ with the nucleus sitting at the origin $x=0$ and $x=R$, respectively. This now directly gives
\begin{align}
    \eqref{eq:proof_propone}&\geq 
    \epsilon \bigg(\norm{\phi_1}_{L^2}^2 + \norm{\phi_2}_{L^2}^2 \bigg) + o(1)
    =  \epsilon \norm{\phi}_{L^2}^2 + o(1),
\end{align}{}
where we used that the lowest eigenvalue $\epsilon = I^{H}$ of $h(x)$ does not depend on the position of the single nucleus in the system.

Combining this lower bound with \eqref{eq:tensor_structure} directly gives the desired assertion.

\clearpage

\printunsrtglossary[type=symbols,style=long]

\appendix

\section{Calculating the exact one-dimensional DFT energy}
For the readers convenience we sketch the calculations for the exact ground state energy of the one-dimensional DFT model (see Section \ref{sec:one_dim}) with the ground-state first reported in \cite{witthaut_1d_exactsolution}.
We consider the corresponding energy functional from \eqref{eq:one-dim_dft}
\begin{align*}
    \E^H[\rho] =     \frac{1}{2} \int \big( \sqrt{\rho}'\big)^2 \di x
    - \rho(0)
    + \big(\tfrac{1}{2} - c_{xc}\big) \int \rho^2 \di x
\end{align*}
with ground-state given by \eqref{eq:1d_wavefunction}
\begin{align} 
   \rho = \alpha |\psi|^2, \quad \text{with} \quad \psi(x) = a \cdot \mathrm{sech}\big( b |x| + x_0\big),
\end{align}{}
where the parameters $a,b,x_0$ only depend on $\alpha$ and $c_{xc}$ and are given by
\[
x_0 =  \mathrm{arctanh}\big(  \frac{1}{b}\big),\quad
a = \sqrt{\frac{b^2}{2(b-1)}}, \quad
b = 1 - \alpha \frac{1-2c_{xc}}{2}.
\]
Note that we are interested in the case $c_{xc} > \frac{1}{2}$, thus $b>0.$
We start off by recalling some basic properties of the $\mathrm{sech}$-Funktion
\[
\frac{\di}{\di x} \mathrm{sech}(x) = - \mathrm{sech}(x) \cdot \mathrm{tanh}(x)
\]
and 
\[
\int \mathrm{sech}(x)^4 \di x = \frac{2}{3} \tanh(x)+ \frac{1}{3} \tanh(x) \sech(x)^2,
\qquad
\int \mathrm{sech}(x)^2  \tanh(x)^2 \di x = \frac{1}{3} \tanh(x)^3.
\]
This implies for the last term in the energy functional
\begin{align*}
     \int \rho^2 \di x 
     &=
     \alpha^2
     a^4 \int_{-\infty}^\infty \sech(b|x| +x_0)^4 \di x\\
     &= 
     \alpha^22a^4\frac{1}{b} \int_{x_0}^{ \infty }\sech(y)^4 \di y
     \\
     &=
     \alpha^2\frac{2a^4}{b} \bigg( \frac{2}{3}   - \frac{2}{3} \tanh ( x_0)  - \frac{1}{3} \tanh(x_0) \sech(x_0)^2  \bigg)\\
     &=
     \alpha^2 \frac{2a^4}{b} \bigg( \frac{2}{3}   - \frac{2}{3b}   - \frac{1}{3b}  (1-\frac{1}{b^2})  \bigg) \\
      &= \frac{\alpha^2}{6} (2b+1).
\end{align*}

For the kinetic energy we obtain
\begin{align*}
    \int \big( \sqrt{\rho}'\big)^2 \di x
    &=
    \alpha a^2b^2 \int_{-\infty}^\infty \sech(b|x|+x_0)^2 \tanh(b|x|+x_0)^2 \di x\\
    &=
    2 \alpha a^2b \int_{x_0}^{\infty } \sech(y)^2\tanh(y)^2 \di y \\
    &= 
    2 \alpha a^2b \frac{1}{3} \bigg( 1 - \tanh( x_0)^3 \bigg)
    =
    2 \alpha a^2b \frac{1}{3} \bigg( 1 - \frac{1}{b^3} \bigg)\\
    &= \frac{\alpha }{3} \big( b^2+b+1\big).
\end{align*}

For the term $\rho(0)$ we have
\[
\rho(0) = \alpha a^2 \sech(x_0)^2 = \alpha  a^2 \big(1 - \frac{1}{b^2} \big) = \alpha  \frac{b+1}{2}.
\]
Therefore we obtain for the total energy
\[
I^H_\alpha =  \E^H[\rho_\alpha] = 
 \frac{\alpha}{6} \big( b^2+b+1\big) - \alpha  \frac{b+1}{2} - (\tfrac{1}{2} - c_{xc}) \frac{\alpha^2}{6} (2b+1)
\]
Plugging now $b = 1- \alpha (\tfrac{1}{2} - c_{xc})$ into this equation gives now the desired result
\[
I^H_\alpha + I^H_{2- \alpha} = \frac{1}{12} (\alpha^2 (3 - 12 c_{xc}^2) + 6 \alpha ( 4 c_{xc}^2-1) -
   4 (1 + 2 c_{xc} + 4 c_{xc}^2)).
\]



\section*{Acknowledgements}

The authors thank Gero Friesecke for helpful discussions and the anonymous reviewers for their helpful comments which helped improve and clarify this manuscript.\\
Both B.R.G. and S.B. gratefully acknowledge support from the International Research Training Group
IGDK Munich - Graz funded by the Deutsche Forschungsgemeinschaft (DFG, German Research Foundation) - Projektummer 188264188/GRK1754.

\addcontentsline{toc}{section}{References}

\bibliographystyle{plain}
\bibliography{density.bib}
\nocite{*}

\end{document}